%% file: article.tex
\documentclass{llncs}

\usepackage[T1]{fontenc}
\usepackage[pdftex]{hyperref}
\usepackage[utf8]{inputenc}
\usepackage{lmodern}

\usepackage{amsmath}
\usepackage{amssymb}
\usepackage{array}
\usepackage{color}
\usepackage{framed}
\usepackage{ifthen}
\usepackage{stmaryrd}
\usepackage{xspace}

\newtheorem{fact}{Fact}

\input{newcommands}

\begin{document}

\title{Pushdown Systems for Monotone Frameworks\thanks{The research presented in
    this paper has been supported by MT-LAB, a VKR Centre of Excellence for the
  Modelling of Information Technology.}
}

\author{Micha{\l} Terepeta
   \and Hanne Riis Nielson
   \and Flemming Nielson
   }
\institute{Technical University of Denmark \\
           \email{\{mtte,riis,nielson\}@imm.dtu.dk}}

\maketitle

\begin{abstract}
  \emph{Monotone frameworks} is one of the most successful frameworks for
  intraprocedural data flow analysis extending the traditional class of
  \emph{bitvector frameworks} (like live variables and available expressions).
  \emph{Weighted pushdown systems} is similarly one of the most general
  frameworks for interprocedural analysis of programs.
  However, it makes use of \emph{idempotent semirings} to represent the sets of
  properties and unfortunately they do not admit analyses whose transfer
  functions are not strict (e.g., classical bitvector frameworks).

  This motivates the development of algorithms for backward and forward
  reachability of pushdown systems using sets of properties forming so-called
  \emph{flow algebras} that weaken some of the assumptions of idempotent
  semirings.
  In particular they do admit the bitvector frameworks, monotone frameworks, as
  well as idempotent semirings.
  We show that the algorithms are \emph{sound} under mild assumptions on the
  flow algebras, mainly that the set of properties constitutes a join
  semi-lattice, and \emph{complete} provided that the transfer functions are
  suitably distributive (but not necessarily strict).
\end{abstract}

%
%

\section{Introduction}
\label{sec:introduction}
\input{introduction}

\section{Monotone Frameworks, Semirings and Flow Algebras}
\label{sec:monotone-fa}
\input{monotone-fa}

\section{Pushdown Systems}
\label{sec:pushdown-fa}
\input{pushdown-fa}

\section{Algorithms}
\label{sec:algorithms}
\input{algorithms}

\section{Soundness}
\label{sec:soundness}
\input{soundness}

\section{Completeness}
\label{sec:completeness}
\input{completeness}

\section{Discussion and examples}
\label{sec:discussion}
\input{discussion}

\section{Conclusions}
\label{sec:conclusions}
\input{conclusions}

%
%

\bibliographystyle{splncs}
\bibliography{article}

\newpage
\appendix

\section{Soundness proofs}
\label{app:ann/soundness-proofs}
\input{soundness-proofs.tex}

\section{Continuity proof (Lem.~\ref{lem:continuous})}
\label{app:ann/continuity-proof}
\input{continuity-proof}

\section{Completeness proofs}
\label{app:ann/completeness-proofs}
\input{completeness-proofs.tex}

\end{document}

%% file: newcommands.tex
\newcommand{\ProofEx}[1]{\text{[ #1 ]}}

\newcommand{\Powerset}{\mathcal{P}}
\newcommand{\Naturals}{\mathbb{N}}

\newcommand{\Pds}{\mathcal{P}}
\newcommand{\ControlLocations}{P}
\newcommand{\StackAlphabet}{\Gamma}
\newcommand{\StackAlphabetS}{\StackAlphabet^{*}}
\newcommand{\PdsRules}{\Delta}
\newcommand{\PdsRulesPre}{\PdsRules_{\textit{pre}}}
\newcommand{\PdsRulesPost}{\PdsRules_{\textit{post}}}
\newcommand{\PdsRulesPostP}{\PdsRules_{\textit{post-2}}}

\newcommand{\PdsRule}{\hookrightarrow}

\newcommand{\PdsTrans}[1]{\overset{#1}{\Longrightarrow}}
\newcommand{\PdsTransS}[1]{\overset{#1}{\Longrightarrow}{}\negthickspace^{*}}

\newcommand{\PreStar}{\MboxIt{Pre}^{*}}
\newcommand{\PostStar}{\MboxIt{Post}^{*}}

\newcommand{\APreStar}{\Automaton_{pre^{*}}}
\newcommand{\APreStarC}{\Automaton_{pre^{*}}^{\Constraints}}

\newcommand{\APostStar}{\Automaton_{post^{*}}}
\newcommand{\APostStarC}{\Automaton_{post^{*}}^{\Constraints}}

\newcommand{\Lhs}{\textit{lhs}}

\newcommand{\Set}[1]{\{ #1 \}}

\newcommand{\Ang}[1]{\langle #1 \rangle}

\newcommand{\PdsA}{\mathcal{PA}}
\newcommand{\ArPds}{\mathcal{A}^{R}\mathcal{P}}
\newcommand{\WPds}{\mathcal{W}}

\newcommand{\Automaton}{\mathcal{A}}
\newcommand{\AutomatonC}{\Automaton^{\Constraints}}

\newcommand{\States}{Q}
\newcommand{\Transitions}{\delta}
\newcommand{\FinalStates}{F}

\newcommand{\Constraints}{\mathcal{C}}

\newcommand{\AutoTransSet}{\boldsymbol{\rightarrow}}

\newcommand{\AutoTrans}[1]{\xrightarrow{#1}}
\newcommand{\AutoTransS}[1]{\AutoTrans{#1}{}\negthickspace^{*}\,}

\newcommand{\AutoTransR}[1]{\xleftarrow{#1}}
\newcommand{\AutoTransRS}[1]{\,{}^{*}\negthickspace\AutoTransR{#1}}
\newcommand{\AutoTransRE}[1]{\overset{#1}{\dashleftarrow}}

\newcommand{\PreConSymb}{l}
\newcommand{\PreCon}[3]{\PreConSymb(#1 \AutoTrans{#2} #3)}

\newcommand{\PreSolSymb}{\lambda}
\newcommand{\PreSolSymbP}{\PreSolSymb^{*}}
\newcommand{\PreSol}[3]{\PreSolSymb(#1 \AutoTrans{#2} #3)}
\newcommand{\PreSolC}[4]{\PreSolSymb#1(#2 \AutoTrans{#3} #4)}
\newcommand{\PreSolP}[4]{\PreSolSymb^{*}(#1 \underset{#3}{\AutoTrans{#2}}{}^{*}\; #4)}

\newcommand{\PostConSymb}{h}
\newcommand{\PostConSymbE}{\PostConSymb^\epsilon}
\newcommand{\PostCon}[3]{\PostConSymb(#3 \AutoTransR{#2} #1)}
\newcommand{\PostConPE}[4]{\PostConSymbE(#4 \underset{#3}{\AutoTransRE{#2}} #1)}

\newcommand{\PostSolSymb}{\lambda}
\newcommand{\PostSolSymbP}{\PostSolSymb_{R}^{*}}
\newcommand{\PostSol}[3]{\PostSolSymb(#3 \AutoTransR{#2} #1)}
\newcommand{\PostSolP}[4]{\PostSolSymb_{R}^{*}(#4 \;^{*}\underset{#3}{\AutoTransR{#2}} #1)}
\newcommand{\PostSolC}[4]{\PostSolSymb#1(#4 \AutoTransR{#3} #2)}

\newcommand{\Semiring}{\mathcal{S}}
\newcommand{\FlowAlgebra}{\mathcal{F}}

\newcommand{\Extend}{\otimes}
\newcommand{\Combine}{\oplus}
\newcommand{\BigCombine}{\bigoplus}
\newcommand{\One}{\bar{1}}
\newcommand{\Zero}{\bar{0}}

\newcommand{\RevComp}{\fatsemi}
\newcommand{\Id}{\textit{id}}

\newcommand{\Leq}{\sqsubseteq}

\newcommand{\BigLub}{\bigsqcup}

\newcommand{\FunSpace}{\mathcal{F}}

\newcommand{\BigFunLub}{\stackrel{\circ}{\BigLub}}

\newcommand{\MboxIt}[1]{\textit{#1}}

%% file: introduction.tex
\emph{Monotone frameworks}~\cite{bib:monotone-frameworks} is a unifying approach
to static analysis of programs.
It creates a generic foundation for specifying various analyses and by imposing
very modest requirements can accommodate a wide range of analyses, including the
bitvector frameworks as well as more complex ones such as constant propagation.
However, the original formulation was focused on the intraprocedural setting and
did not discuss the interprocedural one.

Interprocedural analysis has always been an interesting challenge for static
analysis.
Two of the main reasons for that are the unbounded stack and recursive (or
mutually recursive) procedures.
Moreover, only some paths in the interprocedural flow graph are valid --- the
call and returns should match.
All of this opens up many possibilities for various trade-offs, such as taking
into account or ignoring the calling context.
In their seminal work Sharir and Pnueli~\cite{bib:sharir-pnueli} presented two
approaches allowing for precise interprocedural analysis.
One of them, known as the \emph{call-strings} approach, is based on ``tagging''
the analysis information with the current call stack.
Obviously the length of call-strings should be limited to some threshold in
order to ensure the termination of the analysis.
However, in this paper we will be more interested in the other presented approach.
It is called the \emph{functional} approach and is based on the idea of
computing the summarizations of procedures, i.e., establishing the relationships
between the inputs and the outputs of the blocks of the program and procedures
(composing the results for the blocks).
A similar idea, from the abstract interpretation perspective, was explored
in~\cite{bib:cousot-recursive}, which considered predicate transformers as the
basis for the analysis and also involved constructing systems of functional
equations.

\emph{Pushdown systems}~\cite{bib:pushdown-reachability,bib:efficient-pushdown,bib:schwoon-thesis}
are one of the more recently proposed approaches to interprocedural analysis.
One of the underlying ideas behind them is to use a construction similar to
pushdown automata to model the use of the stack by a program.
An interesting advantage of the approach is the ability to compute the
(possibly) infinite sets of predecessor and successor configurations for a given
program and some initial configurations.
Since the pushdown systems can only handle programs with finite abstractions,
they have been extended with semiring weights/annotations in \emph{weighted}
pushdown systems~\cite{bib:wpds1,bib:wpds2,bib:program-analysis-wpds} and
\emph{communicating} pushdown systems~\cite{bib:cpds1,bib:cpds2}.
The extensions proposed in both of these approaches are actually very close,
although the former focuses on dataflow analysis and generalizing the functional
approach to interprocedural analysis, while the latter on the abstractions of
language generated by synchronization actions in a concurrent setting.
Pushdown systems have been used for verification purposes in many different
projects and contexts.
The examples include the Moped~\cite{bib:schwoon-thesis} and
jMoped~\cite{bib:jmoped} model checkers that extensively use pushdown systems or
Codesurfer~\cite{bib:codesurfer} that takes advantage of weighted pushdown
systems.

However, both the WPDS and CPDS use semiring structure for analysis purposes and
therefore exclude many classical approaches, such as bitvector frameworks where
the transfer functions are not strict.
In this paper we are bringing the pushdown systems based analysis closer to the
monotone frameworks.
To achieve that we use the concept of \emph{flow algebra}~\cite{bib:reykjavik}
that is a structure similar to semiring, but more permissive.
In particular we do not impose the annihilation requirement, nor the
distributivity.
This allows us to present examples of classical analyses that thanks to our
extensions are admitted by the framework, and did not directly fit into the
previous semiring-based approaches.\footnote{Although it is possible to sidestep
this problem by introducing an ``artificial'' annihilator to the semiring.}
Since the existing algorithms are based on the assumption of working with
semiring structure, we develop our slightly different algorithms that allow us
to relax the requirements.
Then we go on to establish the soundness result, i.e., the analysis result
safely over-approximates the join over all valid paths of the pushdown system.
Furthermore, we also prove the completeness of the analysis, that is, provided
that the flow algebra satisfies certain additional properties the result of the
analysis will coincide with the join over all valid paths.

The structure of the paper is as follows.
In Sec.~\ref{sec:monotone-fa} we recall and introduce the necessary concepts,
e.g., monotone frameworks, pushdown systems including both the weighted and
communicating variants.
Then in Sec.~\ref{sec:pushdown-fa} we present basic definitions, while in
Sec.~\ref{sec:algorithms} we describe our algorithms and provide some intuition
behind them.
Then Sec.~\ref{sec:soundness} presents the soundness result for both the forward
and backward reachability.
Similarly Sec.~\ref{sec:completeness} describes the completeness result for both
of them.
Finally, we discuss the results and provide some examples in
Sec.~\ref{sec:discussion} and conclude in Sec.~\ref{sec:conclusions}.

%% file: monotone-fa.tex
In this section we will present the basic definitions that will be used
throughout the rest of the paper.
We will start with recalling the classical approach to static analysis known as
monotone frameworks~\cite{bib:monotone-frameworks,bib:ppa}.
Here we present a slightly more convenient (in the context of this paper)
definition of monotone framework.
\begin{definition}
  \label{def:monotone-framework}
  A complete monotone framework is a tuple
  \[
    (L, \BigLub, \FunSpace, \circ, \Id, (f_l)_{l \in L} )
  \]
  where $L$ is a complete lattice, $\BigLub$ is its least upper bound operator.
  We use $\FunSpace$ to denote a monotone function space on $L$, i.e., a set of
  monotone functions that contains the identity function and is closed under
  function composition.
  Finally, $\circ$ is function composition, $\Id$ is the identity function and
  $f_l = \lambda l' . l$ for every $l \in L$.
\end{definition}
We will also discuss bitvector frameworks, which are a special case of monotone
frameworks.
The lattice used is $L = \Powerset(D)$ for some finite set $D$, the ordering is
either $\subseteq$ or $\supseteq$ and the least upper bound is either $\cup$ or
$\cap$ and the monotone and distributive function space is defined as
\[
  \Set{ f : \Powerset(D) \rightarrow \Powerset(D) \mid
        \exists Y_f^1, Y_f^2 \subseteq D : \forall Y \subseteq D :
        f(Y) = (Y \cap Y_f^1) \cup Y_f^2
  }
\]
One of the main reasons for distinguishing them is the fact that they can be
implemented very efficiently using bitvectors and include common analyses such
as live variables, available expressions, reaching definitions, etc.

Since both weighted and communicating pushdown systems are using semirings, we
will introduce some of the basic definitions associated with
them~\cite{bib:droste-kuich}, starting with the definition of a monoid.
\begin{definition}
  \label{def:monoid}
  A monoid is a tuple $(M, \Extend, \One)$ such that $M$ is non-empty, $\Extend$
  is an associative operator on $M$ and $\One$ is a neutral element for
  $\Extend$, i.e.,
  \[
    \forall a \in M : a \Extend \One = \One \Extend a = a
  \]
\end{definition}
A monoid is \emph{idempotent} if $\Extend$ operator is idempotent, that is
\[
  \forall a \in M : a \Extend a = a
\]
Similarly it is \emph{commutative} if the operator is commutative, in which case
we usually use the symbol $\Combine$ to denote it (and also use $\Zero$ for the
neutral element).
\[
  \forall a, b \in M : a \Combine b = b \Combine a
\]
A commutative monoid $(M, \Combine, \Zero)$ is \emph{naturally ordered} if the
relation defined as
\[
  \forall a, b \in M : a \Leq b \iff \exists c \in M : a \Combine c = b
\]
is a partial order.
Moreover, if the monoid is idempotent then it is naturally ordered and we have
that
\[
  \forall a, b \in M : a \Leq b \iff a \Combine b = b
\]
and $\Combine$ is the least upper bound operator.
Note that this corresponds to a join semi-lattice.

Now we are ready do define the semiring structure.
\begin{definition}
  \label{def:semiring}
  A semiring is a tuple $(S, \Combine, \Extend, \Zero, \One)$ such that
  \begin{itemize}
    \item $(S, \Combine, \Zero)$ is a commutative monoid (hence $\Zero$ is a
      neutral element for $\Combine$)
    \item $(S, \Extend, \One)$ is a monoid (hence $\One$ is a
      neutral element for $\Extend$)
    \item $\Extend$ distributes over $\Combine$, that is
      \begin{align*}
        a \Extend (b \Combine c) & = (a \Extend b) \Combine (a \Extend c)
        \\
        (a \Combine b) \Extend c & = (a \Extend c) \Combine (b \Extend c)
      \end{align*}
    \item $\Zero$ is an annihilator for $\Extend$, that is
      $a \Extend \Zero = \Zero \Extend a = \Zero$
  \end{itemize}
\end{definition}
Similarly to the above, we call a semiring idempotent if $\Combine$ is
idempotent, and commutative if $\Extend$ is commutative.
The ordering for idempotent semiring is defined in the same way as for
idempotent and commutative monoids, with the additional requirement that
$\Extend$ preserves the order (i.e., is monotonic).

As already mentioned we will use the notion of a flow
algebra~\cite{bib:reykjavik}, which is similar to idempotent semirings, but less
restrictive.\footnote{The name comes from the idea of performing
data\textbf{flow} analyses using an \textbf{algebra}ic structure.}
The main difference is that flow algebras do not require the distributivity and
annihilation properties.
Instead we replace the first one with a monotonicity requirement and dispense
with the second one.
It is formally defined as follows.
\begin{definition}
  \label{def:flow-algebra}
  A flow algebra is a structure of the form
  $(F, \Combine, \Extend, \Zero, \One)$ such that:
  \begin{itemize}
    \item $(F, \Combine, \Zero)$ is an idempotent and commutative monoid
    \item $(F,\Extend,\One)$ is a monoid
    \item $\Extend$ is monotonic in both arguments, that is:
      \begin{align*}
        f_1 \Leq f_2 \ &\Rightarrow\ f_1\Extend f\Leq f_2\Extend f \\
        f_1 \Leq f_2 \ &\Rightarrow\ f\Extend f_1\Leq f\Extend f_2
      \end{align*}
  \end{itemize}
  where $ f_1 \Leq f_2 $ if and only if $ f_1 \Combine f_2 = f_2 $.
\end{definition}
Clearly in a flow algebra all finite subsets $ \Set{ f_1, \cdots, f_n } $
have a least upper bound, which is given by
$ \Zero \Combine f_1 \Combine \cdots \Combine f_n $.

Since the assumptions on a flow algebra are less demanding than in the case of
semirings, we additionally introduce the notions of distributive and strict flow
algebras.
\begin{definition}
  \label{def:ditributive-strict-flow-algebra}
  A distributive flow algebra is a flow algebra
  $ (F, \Combine, \Extend, \Zero, \One) $,
  where $\Extend$ distributes over $\Combine$ on both sides, i.e.,
  \begin{align*}
    f_1 \Extend (f_2 \Combine f_3) & = (f_1 \Extend f_2) \Combine (f_1 \Extend f_3)
    \\
    (f_1 \Combine f_2) \Extend f_3 & = (f_1 \Extend f_3) \Combine (f_2 \Extend f_3)
  \end{align*}
  We also say that a flow algebra is strict if
  \[
    \Zero \Extend f = \Zero = f \Extend \Zero
  \]
\end{definition}
\begin{fact}
  Every idempotent semiring is a strict and distributive flow algebra.
\end{fact}
One of the motivations of flow algebras is that the classical bit-vector
frameworks~\cite{bib:ppa} are not strict; hence they are not directly
expressible using idempotent semirings.
Therefore, from this perspective the flow algebras are closer to Monotone
Frameworks, and other classical static analyses.
Restricting our attention to semirings rather than flow algebras would mean
restricting attention to strict and distributive frameworks.
\begin{definition}
  \label{def:reykjavki/complete-fa}
  A complete flow algebra is a flow algebra
  $ (F, \Combine,  \Extend, \Zero, \One) $, where $F$ is a complete lattice; we
  write $\BigCombine$ for the least upper bound.
  It is affine~\cite{bib:ppa} if for all non-empty subsets
  $ F' \not = \emptyset $ of $F$
  \begin{align*}
    f \Extend \BigCombine F' &= \BigCombine \{ f \Extend f' \mid f' \in F' \}
    \\
    \BigCombine F' \Extend f &= \BigCombine \{ f' \Extend f \mid f' \in F' \}
  \end{align*}
  Furthermore, it is completely distributive if it is affine and strict.
\end{definition}
If the complete flow algebra satisfies the ascending chain
condition~\cite{bib:ppa} then it is affine if and only if it is distributive.

Let us emphasize the connection between the flow algebras and the monotone
frameworks.
As defined above a complete monotone framework is
\[
  (L, \BigLub, \FunSpace, \circ, \Id, (f_l)_{l \in L} )
\]
Note that this immediately gives us a flow algebra by taking
\[
  (F, \BigFunLub, \RevComp, f_\bot, id)
\]
where $ \BigFunLub Y = \lambda l . \BigLub_{f \in Y} f(l) $ and
$ f \RevComp g = g \circ f $.

%% file: pushdown-fa.tex
In order to present our results, it is necessary first to introduce some basic
definitions related to pushdown systems as well as weighted/communicating
pushdown systems.

%
%

\subsection{Introduction to Pushdown Systems}

We will start with recalling some of the basic definitions of pushdown systems
and their extensions with semiring weights, namely weighted pushdown systems
(WPDS)~\cite{bib:wpds1,bib:wpds2} and communicating pushdown systems
(CPDS)~\cite{bib:cpds1,bib:cpds2}.
We will mostly follow the notation used for WPDS (note that CPDS use slightly
different notation, but the intent is basically the same in both approaches ---
one equips every pushdown rule with a semiring weight).

\begin{definition}
  \label{def:inter/pds}
  A pushdown system is a tuple
  $ \Pds = (\ControlLocations, \StackAlphabet, \PdsRules) $
  where $P$ is a finite set of control locations, $\Gamma$ is a finite set of
  stack symbols and $\PdsRules$ is a finite set of pushdown rules of the form
  $ \Ang{p, \gamma} \PdsRule \Ang{p', w} $, where $ w \in \Gamma^{*} $ and
  $ |w| \leq 2 $.
\end{definition}
Note that the requirement $ |w| \leq 2 $ is not a serious restriction and any
pushdown system can be transformed to satisfy it.
This can be achieved by adding some fresh control locations and pushing
$|w|$ in a few steps.
The above is already quite enough for checking the reachability of finite
abstractions of programs.
The valuation of global variables can be encoded using control locations
$\ControlLocations$ (and the local variables, if needed, in the stack alphabet
$\StackAlphabet$).

Clearly a pushdown system gives rise to a (possibly infinite) transition
systems, where we can move between configurations using the pushdown rules.
The transition relation for this system is defined more formally below.
For every pushdown rule $ r = \Ang{p_1, \gamma} \PdsRule \Ang{p_2, u} $ we have
\[
  \Ang{p_1, \gamma w} \PdsTrans{r} \Ang{p_2, u w}
\]
for all $ w \in \StackAlphabet^{*} $.
Sometimes we will omit the annotation of the specific pushdown rule --- this
means that we assume there exists a rule that allows moving between the given
configurations.
The reflexive, transitive closure of $\PdsTrans{}$ will be denoted as
$\PdsTransS{}$ (and annotated with sequences of pushdown rules).
Having a precise definition of the transition relation (and its reflexive
transitive closure) allows us to define the concepts of successor and
predecessor configurations.
We call a configuration $c_2$ an immediate successor (predecessor) of $c_1$ if
$ c_1 \PdsTrans{} c_2 $ ($ c_2 \PdsTrans{} c_1 $).
Similar to immediate successors (predecessors) one can also define the general
successors (predecessors) using the $\PdsTransS{}$, namely a configuration $c_2$
a successor (predecessor) of $c_1$ if
$ c_1 \PdsTransS{} c_2 $ ($ c_2 \PdsTransS{} c_1 $).

In many verification problems it is desirable to talk about the sets of
successors or predecessors of a given configuration or set of configurations.
They are often denoted as $\PreStar(C)$ and $\PostStar(C)$ respectively, where
$C$ is some set of configurations.
More formally:
\begin{align*}
  \PreStar(C) &= \Set{ c_2 \mid c_2 \PdsTransS{} c_1, c_1 \in C }
  \\
  \PostStar(C) &= \Set{ c_2 \mid c_1 \PdsTransS{} c_2, c_1 \in C }
\end{align*}
Note that those sets can be in general infinite (even if $C$ is finite).
In order to compute the sets of successors and predecessor we need some symbolic
representation.
Therefore, we define the following.
\begin{definition}
  Given a pushdown system
  $\Pds = (\ControlLocations, \StackAlphabet, \PdsRules)$
  a $\Pds$-automaton is a tuple
  $(\States, \StackAlphabet, \AutoTrans{}, \ControlLocations, \FinalStates)$,
  where:
  \begin{itemize}
    \item $\States$ is a finite set of states such that
      $ \ControlLocations \subseteq \States $
    \item $ \AutoTrans{} \subseteq \States \times \StackAlphabet \times \States $
      is a finite set of transitions
    \item $ \ControlLocations \subseteq \States $ is a finite set of initial
      states
    \item $ \FinalStates \subseteq \States $ is a finite set of final states
  \end{itemize}
\end{definition}
We denote the transitive closure of $\AutoTrans{}$ as $\AutoTransS{}$.
Then we say that a $\Pds$-automaton accepts a configuration $\Ang{p, s}$ if and
only if $p \AutoTransS{w} q$ where $q \in F$.
Moreover, a set of configurations is regular if it is accepted by some
$\Pds$-automaton.

One of the crucial results in the pushdown systems says that the sets of
successors or predecessors of a regular set of configurations are regular
themselves~\cite{bib:pushdown-reachability,bib:efficient-pushdown,bib:schwoon-thesis}.
This is essential since it guarantees that we can always represent those sets as
$\Pds$-automata.
Therefore, the algorithms for $\PreStar$ and $\PostStar$ take as input a
pushdown system and an initial automaton $\Automaton$ that represents the set of
configurations whose predecessors or successors we want to compute.
Both algorithms are basically saturation procedures, i.e., they keep adding new
transitions to the $\Automaton$ according to some rule until no further
transitions (or constraints) can be added.
Since the number of possible transitions is finite (in $\PreStar$ the algorithm
does not add any new states, and in $\PostStar$ always a bounded numer of them),
the algorithms must terminate and return the $\APreStar$ or $\APostStar$, which
represent the possibly infinite number of reachable configurations.

%
%

\subsection{Weighted and Communicating Pushdown Systems}

This approach requires that the sets $\ControlLocations$ and $\StackAlphabet$
are finite, which makes it impossible to use infinite abstractions.
To make it possible to use such abstractions, the
papers~\cite{bib:wpds1,bib:wpds2,bib:cpds1,bib:cpds2} equipped every pushdown
rule with
a semiring value.
As already mentioned we will mostly follow the notation from WPDS, and thus we
present its slightly modified definition below.
\begin{definition}
  \label{def:wpds}
  A weighted pushdown system a tuple $ \WPds = (\Pds, \Semiring, f) $, where
  $\Pds$ is a pushdown system, $ \Semiring = (S, \Combine, \Extend, \Zero, \One) $
  is an idempotent flow algebra and $ f : \PdsRules \rightarrow S $ maps
  pushdown rules to the elements of $S$.
\end{definition}
The main difference when compared to the original definition is that we require
a flow algebra instead of bounded and idempotent semiring.\footnote{Bounded is
used to mean that it contains no infinite ascending
chains~\cite{bib:wpds1,bib:wpds2}.}
Now we can use the fact that every pushdown rule has a flow algebra weight to
define the weight of a sequence of pushdown rules.
Let $ \sigma = [r_1, \ldots, r_n] \in \PdsRules^{*} $ be such a sequence, then
we define $ v(\sigma) = f(r_1) \Extend \cdots \Extend f(r_n) $.
Moveover, the papers extended the algorithms for $\PreStar$ and $\PostStar$ (in
slightly different ways in case of WPDS and CPDS) to handle the addition of
weights.
The result is that both the $\APreStar$ and $\APostStar$ return weighted NFAs,
i.e., where each transition is annotated with a weight.
Apart from making it possible to answer reachability queries, they also provide
additional dataflow information for the given configuration.
In other words, we can not only ask whether a configuration is a successor or
predecessor but also what is the flow algebra value of getting from that
configuration ($\PreStar$) or to that configuration ($\PostStar$).
More formally, we additionally compute the following information:
\begin{itemize}
  \item in case of predecessors of some regular set of configurations $C$ (i.e.,
    if $c_1$ is a predecessor of some configuration in $C$)
    \[
      \delta(c_1) = \BigCombine
        \Set{ v(\sigma) \mid c_1 \PdsTransS{\sigma} c_2, c_2 \in C }
    \]
    is the flow algebra value of all the paths going from configuration
    $c_1 = \Ang{p, s}$ ($s \in \StackAlphabetS$) to any configuration in $C$.
    It can be obtained by simulating $\APreStar$ from state $p$ with input $s$
    multiplying the weights of the transitions in the same order as they are
    taken.
  \item in case of successors of some regular set of configurations $C$ (i.e.,
    if $c_1$ is a successor of some configuration in $C$)
    \[
      \delta(c_1) = \BigCombine
        \Set{ v(\sigma) \mid c_2 \PdsTransS{\sigma} c_1, c_2 \in C }
    \]
    is the flow algebra value of all the paths going from any configuration in
    $C$ to $c_1 = \Ang{p, s}$ ($s \in \StackAlphabetS$).
    It can be obtained by simulating $\APostStar$ from state $p$ with input $s$
    multiplying the weights of the transitions in the reverse order as they are
    taken.
\end{itemize}
Note that in both cases we only want to calculate the value for a predecessor or
successor, thus the sets of paths are never empty.
In case of $\PostStar$, the intuition behind reading the weights of a path in
the automaton in the reverse order is that when a configuration
$ \Ang{p, \gamma_k \ldots \gamma_1} $ is accepted, this means that there are
transitions in the automaton such that the first one is labeled with $\gamma_k$,
the second with $\gamma_{k-1}$ and so on.
However, when one thinks how the program would actually execute, it would build
the stack from the other end, i.e before it can push $\gamma_2$ on the stack, it
must push $\gamma_1$.
Therefore, the weights should be multiplied in the reverse order.

%% file: algorithms.tex
As already mentioned, WPDS and CPDS are assuming that the abstract domain forms
a semiring structure.
This immediately excludes standard analyses based on monotone framework or
bitvector framework.
Fortunately we will show that it is possible to formulate algorithms for
$\PreStar$ and $\PostStar$ that do not need this assumption.
We achieve that by generating the constraints during the saturation procedures
that create the $\APreStar$ and $\APostStar$ (in WPDS no constraints are
generated and the weights are calculated directly, in CPDS constraints are
generated independently of the $\APreStar$ and $\APostStar$ construction).
We will use $\APreStarC$ and $\APostStarC$ to denote the automata with the
associated set of constraints $\Constraints$.
The rest of the section will introduce the algorithms and in the subsequent
sections we will discuss their soundness and completeness.
In this way we believe that we can present the minimum requirements that are
necessary for interprocedural analysis based on pushdown systems.

%
%

\subsection{Algorithm for $\PreStar$}

The procedure introduced in this section is quite similar to the one
from~\cite{bib:cpds1,bib:cpds2} as it generates explicit constraints.
However, it does it during the automaton computation not separately.
In this respect it is somewhat similar to the procedure
from~\cite{bib:wpds1,bib:wpds2} that computes both the weights and the automaton
at the same time.
Also, note that there is no difference with respect to how the new transitions
are added to the automaton.
Therefore, we are able to reuse the standard results with respect to the
automaton itself (i.e., excluding the weights).

The algorithm is as follows.
First, for every transition $q \AutoTrans{\gamma} q'$ in $\Automaton$ we add a
constraint
\[
\One \Leq \PreCon{q}{\gamma}{q'}
\]
(we use $\PreConSymb(-)$ in the constraints to denote the weight of the given
transition)
Then we perform the saturation procedure on $\Automaton$ along with the
generation of constraints that are added to $\Constraints$.
For every pushdown rule $r$ in $\PdsRules$:
\begin{itemize}
\item if $r = \Ang{p, \gamma} \PdsRule \Ang{p', \epsilon}$
  we add a transition
  \[
  p \AutoTrans{\gamma} p'
  \]
  along with the following constraint
  \[
  f(r) \Leq \PreCon{p}{\gamma}{p'}
  \]
\item if $r = \Ang{p, \gamma} \PdsRule \Ang{p', \gamma'}$
  and there is a transition $p' \AutoTrans{\gamma'} q$ in the current
  automaton, we add a transition
  \[
  p \AutoTrans{\gamma} q
  \]
  along with the following constraint
  \[
  f(r) \Extend \PreCon{p'}{\gamma'}{q} \Leq \PreCon{p}{\gamma}{q}
  \]
\item if $r = \Ang{p, \gamma} \PdsRule \Ang{p', \gamma' \gamma''}$
  and there is a path $p' \AutoTrans{\gamma'} q' \AutoTrans{\gamma''} q$ (for some
  $q'$) in the current automaton, we add a transition
  \[
  p \AutoTrans{\gamma} q
  \]
  along with the following constraint
  \[
  f(r) \Extend \PreCon{p'}{\gamma'}{q'} \Extend \PreCon{q'}{\gamma''}{q}
    \Leq \PreCon{p}{\gamma}{q}
  \]
\end{itemize}
We stop once we cannot add any new constraints or transitions.
And since the number of possible transitions and constraints is finite, the
procedure will always terminate.

%
%

\subsection{Algorithm for $\PostStar$}

As in the case of $\PreStar$ algorithm, we only change the way the constraints
are generated, and not how new transitions are added to the automaton.
Recall that we require the initial automaton $\Automaton$ to have no
transitions going into the initial states nor any $\epsilon$-transitions.
As already mentioned, the weights of a $\APostStar$ should be multiplied in the
reverse order compared to the take transitions.
Therefore, we will use the reverse arrow notation for the transitions of the
automata, i.e., we will write $q \AutoTransR{\gamma} p$ for the transition
earlier denoted by $p \AutoTrans{\gamma} q$.
Furthermore, as already noted in~\cite{bib:cpds1,bib:cpds2} the
$\epsilon$-transitions added by the algorithm always originate in an initial
state and go only to some non-initial state.
Therefore, we can conclude that we can take at most one $\epsilon$-transition
(when going from initial state to some non-initial one) and then we can only
take non $\epsilon$-transitions.
Therefore, let us use $\AutoTransRE{\gamma}$ to denote
$(\AutoTransR{\gamma} \circ \AutoTransR{\epsilon}) \cup \AutoTransR{\gamma}$
and define $\PostConSymbE$ as
\[
\PostConSymbE(\rho) =
  \begin{cases}
  \PostCon{p}{\gamma}{q}
    & \mbox{ if } \rho = q \AutoTransR{\gamma} p \\
  \PostCon{q'}{\gamma}{q} \Extend \PostCon{p}{\epsilon}{q'}
    & \mbox{ if } \rho = q \AutoTransR{\gamma} q' \AutoTransR{\epsilon} p \\
  \end{cases}
\]

The algorithm is as follows.
First, for every transition $q' \AutoTransR{\gamma} q$ in $\Automaton$ we add a
constraint
\[
\One \Leq \PostCon{q}{\gamma}{q'}
\]
Then for all pushdown rules of the form
$\Ang{p, \gamma} \PdsRule \Ang{p', \gamma' \gamma''}$
we add a new state $q_{p', \gamma'}$ to the automaton.
Finally, for every pushdown rule $r$ in $\PdsRules$:
\begin{itemize}
\item if $r = \Ang{p, \gamma} \PdsRule \Ang{p', \epsilon}$
  and there is a path $\rho = q \AutoTransRE{\gamma} p$ then add a transition
  \[
  q \AutoTransR{\epsilon} p'
  \]
  along with the following constraint
  \[
  \PostConPE{p}{\gamma}{\rho}{q} \Extend f(r) \Leq \PostCon{p'}{\epsilon}{q}
  \]
  Note that this transition (and its weight) takes care of the return from a
  procedure.
\item if $r = \Ang{p, \gamma} \PdsRule \Ang{p', \gamma'}$
  and there is a path $\rho = q \AutoTransRE{\gamma} p$ then add a transition
  \[
  q  \AutoTransR{\gamma'} p'
  \]
  along with the following constraint
  \[
  \PostConPE{p}{\gamma}{\rho}{q} \Extend f(r) \Leq \PostCon{p'}{\gamma'}{q}
  \]
\item if $r = \Ang{p, \gamma} \PdsRule \Ang{p', \gamma' \gamma''}$
  and there is a path $\rho = q \AutoTransRE{\gamma} p$ then add transitions
  \[
  q  \AutoTransR{\gamma''} q_{p', \gamma'}
  \qquad
  q_{p', \gamma'} \AutoTransR{\gamma'} p'
  \]
  along with the following constraints
  \begin{align*}
  \One & \Leq \PostCon{p'}{\gamma'}{q_{p', \gamma'}} \\
  \PostConPE{p}{\gamma}{\rho}{q} \Extend f(r) & \Leq \PostCon{q_{p', \gamma'}}{\gamma''}{q} \\
  \end{align*}
  Note that this transition $q \AutoTransR{\gamma''} q_{p', \gamma'}$ (and its
  weight) takes care of the procedure call.
\end{itemize}
Again, as in the case of $\PreStar$ we stop once we cannot add any new
constraints or transitions.
And since the number of possible transitions and constraints is finite, the
procedure will always terminate (note that we add some new states only at the
beginning of the procedure and not in the saturation phase).

%% file: soundness.tex
In this section we will discuss and present the main results regarding the
soundness of our algorithms.
Since one of the goals of our formulation of the algorithms is to make the
requirements imposed on the abstract domain explicit and precise, we take a
particular approach to the soundness proofs.
We do not discuss how the generated constraints can be solved (and if they can
be solved at all).
Instead we assume that some solution to those constraints is available and show
that it is a safe over-approximation of the join over all valid paths.

Apart from that, separating the requirements necessary to solve the constraints
from the soundness result gives us the flexibility to easily accommodate
different techniques of solving the constraints.
One can use the usual Kleene iteration, but also more recent approaches using
Newton's method generalized to $\omega$-continuous
semirings~\cite{bib:newtonian,bib:derivation-tree-analysis}.
Furthermore, it also makes it clear that techniques such as widening can be used
for domains that contain infinite ascending chains but do not satisfy the
requirements of Newton's method.

%
%

\subsection{$\PreStar$}

We will start with some intuition about how the pushdown system $\Pds$ and the
automaton $\Automaton$ fit together.
Observe that if a configuration is backward reachable from $C$, there exists a
sequence of pushdown rules in $\PdsRules$ such that the resulting configuration
is accepted by $\Automaton$.
Therefore, we can intuitively think about this system as a one big pushdown
system
$\PdsA = (P, \Gamma, \PdsRulesPre)$, where
\[
  \PdsRulesPre = \PdsRules \cup
    \Set{
      \Ang{q, \gamma} \PdsRule{} \Ang{q', \epsilon} \mid
        q \AutoTrans{\gamma} q' \in \AutoTransSet
    }
\]
With each added pushdown rule we associate the weight $\One$.
This system works by first acting like $\Pds$ and then, at some point, switching
to simulating $\Automaton$ (with the added pushdown rules).
Note that once $\PdsA$ starts using the added pushdown rules, it cannot use the
ones of $\Pds$.
This is because rules in $\Pds$ correspond to the initial states of $\Automaton$
and since it does not have any transitions going to initial states, then the
first used rule from $\PdsRulesPre \setminus \PdsRules$ will go to some
non-initial state.
Thus no pushdown rule of $\Pds$ will be applicable.

This is useful because it allows us to look at the problem of predecessors of
$C$ from a slightly different angle.
Let us consider the automaton $\APreStar$, we say that a configuration $c_p$ is a
predecessor of some configuration $c \in C$ if there is a sequence $\sigma \in
\PdsRules^{*}$ of pushdown rules such that $c_p \PdsTransS{\sigma} c$.
But since $c$ is recognized by $\Automaton$ then there is a sequence
$\sigma' \in \PdsRulesPre^{*}$ such that
$c \PdsTransS{\sigma'} \Ang{q_f, \epsilon}$ for some final state $q_f$.
Therefore, an alternative way to define a predecessor is to say that a
configuration $c_p$ is a predecessor of some configuration $c$ in $C$ if there is
a sequence $\sigma_p \in \PdsRulesPre^{*}$ of pushdown rules such that
$c_p \PdsTransS{\sigma_p} \Ang{q_f, \epsilon}$ for some state $q_f \in F$.
Moreover, since we have that each of the added rules has weight $\One$ then
$v(\sigma) = v(\sigma_p)$.

In the following sections the solution to the constraints will be denoted as
$\PreSolSymb$ (i.e., maps each transition to its weight).
Its generalization to paths $\PreSolSymbP$ is inductively defined as follows:
\[
  \PreSolSymbP(\rho) =
  \begin{cases}
    \PreSol{q}{\gamma}{q'}
      & \text{if} \quad
      \rho = q \AutoTrans{\gamma} q' \\
    \PreSol{q}{\gamma}{q''} \Extend \PreSolSymbP(\rho')
      & \text{if} \quad
        \rho = q \AutoTrans{\gamma} \overbrace{q'' \AutoTransS{s'} q'}^{\rho'}
  \end{cases}
\]

Now we are ready to prove that a solution to the constraints generated by our
saturation procedure is sound.
\begin{theorem}
  \label{thm:ann/pre-soundness}
  Consider an automaton $\Automaton$ and its corresponding $\APreStarC$
  generated by the saturation procedure.
  Let us assume that we have a solution $\PreSolSymb$ to the set of constraints
  $\Constraints$.
  Then for each pair $(p, s)$ such that
  $\Ang{p, s} \PdsTransS{\sigma} \Ang{q_f, \epsilon}$
  (where $\sigma \in \PdsRulesPre^{*}$ and $q_f \in F$), we have
  $v(\sigma) \Leq \PreSolSymbP(\rho)$ where
  $\rho = p \AutoTransS{s} q_f$ is in $\APreStar$.
\end{theorem}
\begin{proof}
The proof is available in App.~\ref{app:ann/pre-soundness}.
\end{proof}

%
%

\subsection{$\PostStar$}

As previously we can think about this system as a one big pushdown system.
However, this time such a system would first simulate the reverse of
$\Automaton$, i.e., instead of accepting some configuration, it generates
one; and only then continue by running the pushdown system itself.
Let us denote such a system as
$\ArPds = (P, \Gamma, \PdsRulesPost)$, where $\PdsRulesPost$ is defined as
follows.
\begin{itemize}
\item For every $q' \AutoTransR{\gamma} q$ in $\Automaton$ we have a rule
  $r = \Ang{q', \epsilon} \PdsRule{} \Ang{q, \gamma}$ in $\PdsRulesPost$
  such that $f(r) = \One$.
\item All other rules of $\PdsRules$ are included in $\PdsRulesPost$.
\end{itemize}

Let us consider the automaton $\APostStar$. We say that a configuration $c'$ is a
successor of some configuration $c$ in $C$ if there is a sequence $\sigma \in
\PdsRules^{*}$ of pushdown rules such that
$c \PdsTransS{\sigma} c'$.
But since $c$ is recognized by $\Automaton$ then there is a sequence
$\sigma' \in \PdsRulesPost^{*}$ such that
$\Ang{q_f, \epsilon} \PdsTransS{\sigma'} c$ for some final state $q_f$.
Therefore, an alternative way to define a successor is to say that
a configuration $c'$ is a successor of some configuration $c \in C$ if there
is a sequence $\sigma_p \in \PdsRulesPost^{*}$ of pushdown rules such that
$\Ang{q_f, \epsilon} \PdsTransS{\sigma_p} c'$ for some state $q_f \in F$.
Moreover, since we have that each of the added rules has weight $\One$ then
$v(\sigma) = v(\sigma_p)$.

Similarly as in the case of $\PreStar$, we define $\PostSolSymbP$ in the
following way:
\[
  \PostSolSymbP(\rho) =
  \begin{cases}
  \PostSol{q'}{\gamma}{q}
    & \text{if} \quad
    \rho = q' \AutoTransR{\gamma} q \\
  \PostSolSymbP(\rho') \Extend \PostSol{q''}{\gamma}{q}
    & \text{if} \quad
      \rho = \overbrace{q' \AutoTransRS{s'} q''}^{\rho'} \AutoTransR{\gamma} q
  \end{cases}
\]
As already mentioned we multiply the weight in the reverse order compared to the
order of transitions in the given path.

\begin{theorem}
  \label{thm:ann/post-soundness}
  Consider an automaton $\Automaton$ and its corresponding $\APostStarC$
  generated by the saturation procedure.
  Let us assume that we have a solution $\PostSolSymb$ to the set of constraints
  $\Constraints$.
  Then for each pair $(p, s)$ such that
  $\Ang{q_f, \epsilon} \PdsTransS{\sigma} \Ang{p, s}$
  (where $\sigma \in \PdsRulesPost^{*}$ and $q_f \in F$), we have
  $v(\sigma) \Leq \PostSolSymbP(\rho)$ where
  $\rho = q_f \AutoTransRS{s} p$ is in $\APostStarC$.
\end{theorem}
\begin{proof}
The proof is available in App.~\ref{app:ann/post-soundness}.
\end{proof}

%% file: completeness.tex
In this section we will prove the completeness of our procedure, i.e., we will
show that provided the abstract domain satisfies certain conditions, the
solution to the generated constraints will coincide with the join over all valid
paths.
The presentation of the results (and their proofs) is quite a bit different than
in the case of soundness.
This is mainly due to the additional complexity of the proofs as well as some
additional restrictions that must to be imposed.
Throughout the whole section we assume that the flow algebra is both
\emph{complete} and \emph{affine}.
In other words we have least upper bounds of arbitrary sets and $\Extend$
distributes over sums of all non-empty sets.

Before we present the main results for each of the two algorithms, let us first
establish that the solution to the generated constraints can be obtained by
Kleene iteration.
To achieve that we will define a function that represents the constraints and
show that it is continuous.
Let us recall that all generated constraints are of similar form: the right-hand
side is a variable and the left-hand side is a finite expression mentioning at
most two variables.
The finite expressions are constructed using $\Combine$ and $\Extend$ which are
themselves affine and hence continuous.

For clarity let $\Constraints_t \subseteq \Constraints$ denote the finite set of
the constraints that have the variable $t$ on the right-hand side.
Recall that each variable corresponds to a transition in an automaton.
Similarly we will use $\Lhs_m(c)$ ($c \in \Constraints$) to denote the
interpretation of the left-hand side of the constraint $c$ under the assignment
$m$.

What we want to compute is a mapping $m$ that is a fixed point of:
\begin{align*}
  & F : (\Transitions \rightarrow D) \rightarrow (\Transitions \rightarrow D)
  \\
  & F(m) t = \BigCombine_{c \in \Constraints_t} \Lhs_m(c)
\end{align*}
where $\Transitions$ is the set of all transitions.
\begin{lemma}
  \label{lem:continuous}
  $F$ is continuous, i.e., for any non-empty chain $Y$:
  \[
    F(\BigLub Y) = \BigLub_{m \in Y} F(m)
  \]
\end{lemma}
\begin{proof}
  The proof is available in App.~\ref{app:ann/continuity-proof}.
\end{proof}
It follows that $ \BigLub \Set{ F^{n}(\bot) \mid n \in \Naturals } $ is the
least solution to our constraint system.

%
%

\subsection{$\PreStar$}

We will first establish a lemma showing that every transition in the $\APreStar$
automaton has at least one corresponding path in the $\PdsA$.
This will be useful in subsequent proofs where we need the fact that certain
sets of $\PdsA$ paths are not empty.
\begin{lemma}
  \label{lem:ann/pre-exists-path}
  For every transition $q \AutoTrans{\gamma} q'$ in $\APreStar$ there exists a
  sequence $\sigma \in \PdsRulesPre$ such that
  $\Ang{q, \gamma} \PdsTransS{\sigma} \Ang{q', \epsilon}$.
\end{lemma}
\begin{proof}
  The proof is available in App.~\ref{app:ann/pre-exists-path}.
\end{proof}

First we will establish the essential result for a single transition of the
created automaton.
\begin{lemma}
\label{lem:ann/pre-completeness-transition}
  Consider a weighted pushdown system $\WPds = (\Pds, \FlowAlgebra, f)$ where
  $\FlowAlgebra$ is affine and an automaton $\APreStarC$ created by the
  saturation procedure.
  Moreover, let $\PreSolSymb$ be the least solution to the set of constraints
  $\Constraints$.
  For every transition $q \AutoTrans{\gamma} q'$ in this automaton we have that
  \[
  \PreSol{q}{\gamma}{q'} \Leq \BigCombine \Set{ v(\sigma) \mid
    \Ang{q, \gamma} \PdsTransS{\sigma} \Ang{q', \epsilon}, \sigma \in \PdsRulesPre^{*} }
  \]
\end{lemma}
\begin{proof}
  The proof is available in App.~\ref{app:ann/pre-completeness-transition}.
\end{proof}
This is also the place that we have used the fact that the solution is equal to
the least upper bound of the ascending Kleene sequence.

And now we can generalize the above to the case of a path in the automaton.
\begin{lemma}
  \label{lem:ann/pre-completeness-path}
  Consider a weighted pushdown system $\WPds = (\Pds, \FlowAlgebra, f)$ where
  $\FlowAlgebra$ is affine and a $\APreStarC$ automaton created by the
  saturation procedure.
  Moreover, let $\PreSolSymb$ be the least solution to the set of constraints
  $\Constraints$.
  For every path $\rho = q \AutoTransS{s} q'$ in this automaton we have that
  \[
    \PreSolP{q}{s}{\rho}{q'} \Leq \BigCombine \Set{ v(\sigma) \mid
      \Ang{q, s} \PdsTransS{\sigma} \Ang{q', \epsilon}, \sigma \in \PdsRulesPre^{*}
    }
  \]
\end{lemma}
\begin{proof}
  The proof is available in App.~\ref{app:ann/pre-completeness-path}.
\end{proof}

And finally, using both the Thm.~\ref{thm:ann/pre-soundness} and the above Lemma, we
can formulate the main result.
\begin{theorem}
  \label{thm:ann/pre-completeness}
  Consider an automaton $\APreStarC$ constructed by the saturation procedure and
  let $\PreSolSymb$ be the least solution to the set of its constraints
  $\Constraints$.
  If the flow algebra is affine then for every path
  $\rho = p \AutoTransS{s} q_f$ where $q_f \in F$ we have that
  \[
    \PreSolP{p}{s}{\rho}{q_f} = \BigCombine
    \Set{ v(\sigma) \mid \Ang{p, s} \PdsTransS{\sigma} \Ang{q_f, \epsilon},
                         \sigma \in \PdsRulesPre^{*}
    }
  \]
\end{theorem}
\begin{proof}
  The proof is available in App.~\ref{app:ann/pre-completeness}.
\end{proof}

%
%

\subsection{$\PostStar$}

Consider a pushdown system $\Pds$ with pushdown rules $\PdsRules$ and a regular
set of  configurations $C$ with an automaton $\Automaton$ that accepts $C$.
First let us define a small modification of the pushdown rules $\PdsRules$.
Each rule $r$ of the form
\[
  \Ang{p, \gamma} \PdsRule{} \Ang{p', \gamma_1 \gamma_2}
\]
can be ``split'' into two rules $r_1$ and $r_2$:
\begin{align*}
  r_1 & = \Ang{p, \gamma} \PdsRule{} \Ang{q_{p', \gamma_1}, \gamma_2} \\
  r_2 & = \Ang{q_{p', \gamma_1}, \epsilon} \PdsRule{} \Ang{p', \gamma_1}
\end{align*}
with weights $f(r_1) = f(r)$ and $f(r_2) = \One$.
Note that the second rule is not really a pushdown rule as defined earlier.
Fortunately, all we need to do, is to redefine $\PdsTrans{}$ in the following
way:
\begin{align*}
  & \text{if} \quad r = \Ang{q, \gamma} \PdsRule{} \Ang{q', w}
  && \text{then} &&
  \forall w' \in \Gamma^{*} : \Ang{q, \gamma s} \PdsTrans{} \Ang{q', w s}
  \\
  & \text{if} \quad r = \Ang{q, \epsilon} \PdsRule{} \Ang{q', \gamma}
  && \text{then} &&
  \forall w' \in \Gamma^{*} : \Ang{q, s} \PdsTrans{} \Ang{q', \gamma s}
\end{align*}
This does not change the pushdown system in any way.
Since we add a fresh state, there is no danger of changing any paths except for
the ones we intend to.
Moreover, the weight remains the same ($\One$ is neutral element for $\Extend$,
so $f(r_1) \Extend f(r_2) = f(r)$).

Therefore, in place of $\PdsRulesPost$ we will use $\PdsRulesPostP$, which is
defined as follows:
\begin{itemize}
  \item For every $q' \AutoTransR{\gamma} q$ in $\Automaton$ we have a rule
    $r = \Ang{q', \epsilon} \PdsRule{} \Ang{q, \gamma}$ in $\PdsRulesPostP$
    such that $f(r) = \One$.
  \item For every $r \in \PdsRules$ of the form
    $r = \Ang{p, \gamma} \PdsRule{} \Ang{p', \gamma_1 \gamma_2}$ there are $r_1$
    and $r_2$ in $\PdsRulesPostP$ as described above.
  \item All other rules of $\PdsRules$ are included in $\PdsRulesPostP$
    without any modification.
\end{itemize}
So compared to $\PdsRulesPost$ the only difference is that we split the
push-rules into two separate rules.
At the same time we do not change the behavior of the system in any way.

This allows us to prove the following lemma, which is used in subsequent
proofs.
\begin{lemma}
  \label{lem:ann/post-exists-path}
  For every transition $q' \AutoTransR{\gamma_\epsilon} q$
  ($\gamma_\epsilon \in \Gamma \cup \{ \epsilon \}$)
  in $\APostStar$ there exists a sequence $\sigma$ of pushdown rules in
  $\PdsRulesPostP$ such that
  $\Ang{q', \epsilon} \PdsTransS{\sigma} \Ang{q, \gamma_\epsilon}$.
\end{lemma}
\begin{proof}
  The proof is available in App.~\ref{app:ann/post-exists-path}.
\end{proof}

Again, as in the case of $\PreStar$ we first establish the result for a single
transition in the automaton.
\begin{lemma}
  \label{lem:ann/post-completeness-transition}
  Consider a weighted pushdown system $\WPds = (\Pds, \FlowAlgebra, f)$ where
  $\FlowAlgebra$ is affine and an automaton $\APostStarC$ created by the
  saturation procedure.
  Moreover, let $\PostSolSymb$ be the least solution to the set of constraints
  $\Constraints$.
  For every transition $q' \AutoTransR{\gamma_\epsilon} q$
  ($\gamma_\epsilon \in \Gamma \cup \{ \epsilon \}$)
  in this automaton we have that
  \[
  \PostSol{q}{\gamma_\epsilon}{q'} \Leq \BigCombine \Set{ v(\sigma) \mid
    \Ang{q', \epsilon} \PdsTransS{\sigma} \Ang{q, \gamma_\epsilon}, \sigma \in
    \PdsRulesPostP^{*} }
  \]
\end{lemma}
\begin{proof}
  The proof is available in App.~\ref{app:ann/post-completeness-transition}.
\end{proof}
And again, as in the case of $\PreStar$, this is the place that we have used the
fact that the solution is equal to the least upper bound of the ascending Kleene
sequence.

Now we can generalize the obtained result for the paths in the automaton.
\begin{lemma}
  \label{lem:ann/post-completeness-path}
  Consider a weighted pushdown system $\WPds = (\Pds, \FlowAlgebra, f)$ where
  $\FlowAlgebra$ is affine and a $\APostStarC$ automaton created by the saturation
  procedure.
  Moreover, let $\PostSolSymb$ be the least solution to the set of constraints
  $\Constraints$.
  For every path $\rho = q' \AutoTransR{s} q$
  ($s \in \Gamma^{*}$) in this automaton we have that
  \[
  \PostSolP{q}{s}{\rho}{q'} \Leq \BigCombine \Set{ v(\sigma) \mid
    \Ang{q', \epsilon} \PdsTransS{\sigma} \Ang{q, s}, \sigma \in \PdsRulesPostP^{*} }
  \]
\end{lemma}
\begin{proof}
  The proof is available in App.~\ref{app:ann/post-completeness-path}.
\end{proof}

And finally using both the soundness Thm.~\ref{thm:ann/post-soundness} and the
above, we can establish the main result.
\begin{theorem}
  \label{thm:ann/post-completeness}
  Consider an automaton $\APostStarC$ constructed by the saturation procedure and
  let $\PostSolSymb$ be the least solution to the set of its constraints
  $\Constraints$.
  If the flow algebra is affine then for every path
  $\rho = q_f \AutoTransRS{s} p$ where $q_f \in F$ we have that
  \[
  \PostSolP{p}{s}{\rho}{q_f} = \BigCombine \Set{ v(\sigma) \mid
    \Ang{q_f, \epsilon} \PdsTransS{\sigma} \Ang{p, s}, \sigma \in \PdsRulesPostP^{*}
  }
  \]
\end{theorem}
\begin{proof}
  The proof is available in App.~\ref{app:ann/post-completeness}.
\end{proof}

%% file: discussion.tex
In this section we will discuss the relation of our development to the area of
interprocedural analysis, as well as the challenges and advantages of the
approach.
Furthermore, we will present an example of analyses that thanks to our algorithms
are directly expressible in our framework, which was not possible before.

\subsection{Monotone frameworks and pushdown systems}

To put our approach into perspective, it is useful to emphasize that it is a
generalization of the functional approach to interprocedural analysis by Sharir
and Pnueli~\cite{bib:sharir-pnueli}.
In both of these approaches the underlying idea is to compute the summarizations
of actions and by composing them obtain the summarizations of procedures.
The generality of weighted pushdown systems stems from the fact that they make
it possible to obtain the analysis information for specific calling contexts or
even families of calling contexts.
In other words one can perform queries of weighted $\APreStar$ and $\APostStar$
automata, to get the summarization of all the paths between the initial set of
configurations and a given stack or even a regular set of stacks.
Applying the summarization to some initial analysis information, we can obtain
the desired result.
This is possible due to the way the algorithms for pushdown systems construct
the $\APreStar$ and $\APostStar$ automata and generate the constraints whose
solution provides us with the weights of all the transition in those automata.

One of the most significant advantages of using summarizations is the fact that
each procedure can be analyze only once and the result can be used at all the
call sites.
In other words the summarization of a procedure is independent of the calling
context, which is the key to reusing the information.
However, there is also a downside to this approach, namely the fact that the
analysis has to work on the dataflow transformers and not directly on some
dataflow facts (i.e., we compute what and how the dataflow facts can change).
This often makes it more difficult to formulate analyses whose results we can
actually compute.
The main challenge is that if some domain $D$ satisfies, e.g., the ascending
chain condition, when lifted to transformers $ D \rightarrow D $ it might not
satisfy this condition anymore.
Fortunately we can still express many analyses.
Even for cases like constant propagation where $D$ is usually a mapping from
variables to integers/reals, it is possible to define computable variants, i.e.,
copy- and linear-constant propagation~\cite{bib:wpds1,bib:wpds2}.
Obviously whenever $D$ is finite then $D \rightarrow D$ will be finite as well.
This might seem a bit restrictive, but there are many analyses that satisfy the
requirement.
In fact the interprocedural analysis based on graph
reachability~\cite{bib:reps-graph} works on distributive functions
$\Powerset(D) \rightarrow \Powerset(D)$ where $D$ is required to be some finite
set.

\subsection{Example}

As an example let us consider the family of forward, may analyses that are
instances of bitvector framework.
They are generally defined in the following way:
\begin{itemize}
  \item The lattice $L$ is equal to $\Powerset(D)$ for some finite $D$.
  \item The least upper bound operator is $\bigcup$.
  \item The transfer functions are monotone functions of the shape
    \[
      f_i(l) = (l \setminus k_i) \cup g_i
    \]
    where $ k_i, g_i \in \Powerset(D) $ correspond to the elements of $D$ that
    are ``killed'' and ``generated'' at some program point $i$.
    This is also the source of a popular name for similar analyses ---
    ``kill/gen'' analyses.
  \item The least element $ \bot = \emptyset $.
\end{itemize}
In order to use such an analyses with weighted pushdown systems we will
construct a flow algebra $ (\FlowAlgebra, \Combine, \Extend, \Zero, \One) $
that expresses the transformers
$ \Powerset(D) \rightarrow \Powerset(D) $.
Since we are dealing with ``kill/gen'' analysis, this is actually quite easy ---
we express a function $ f_ i(l) = (l \setminus k_i) \cup g_i $ by a pair
$(k_i, g_i)$.
Therefore, we have:
\begin{itemize}
  \item $ \FlowAlgebra = \Powerset(D) \times \Powerset(D) $
  \item The $\Combine$ operator is defined as
    \[
      f_1 \Combine f_2
      = (k_1, g_1) \Combine (k_2, g_2)
      = (k_1 \cap k_2, g_1 \cup g_2)
    \]
  \item The $\Extend$ operator is defined as
    \[
      f_1 \Extend f_2
      = (k_1, g_1) \Extend (k_2, g_2)
      = (k_1 \cup k_2, (g_1 \setminus k_2) \cup g_2)
    \]
  \item $ \Zero = (D, \emptyset) $
  \item $ \One = (\emptyset, \emptyset) $
\end{itemize}
It should be easy to see that $\Combine$ is idempotent and commutative.
Therefore, the semiring is naturally ordered with
$ f_1 \Leq f_2 \iff f_1 \Combine f_2 = f_2 $.
Furthermore, $\Zero$ is a neutral element for $\Combine$ and $\One$ is neutral
for $\Extend$.

However, the interesting part is that $\Zero$ is not an annihilator for
$\Extend$.
Consider the following:
\begin{align*}
  (D, \emptyset) \Extend (k, g)
  & = (D \cup k, (\emptyset \setminus k) \cup g)
  \\
  & = (D, g)
\end{align*}
which clearly is not equal to $\Zero$ (unless $ g = \emptyset $).
Interestingly the annihilation works from the right:
\begin{align*}
  (k, g) \Extend (D, \emptyset)
  & = (k \cup D, (g \setminus D) \cup \emptyset)
  \\
  & = (D, \emptyset)
  \\
  & = \Zero
\end{align*}
This makes perfect sense if we consider for a moment the classical transfer
functions of such analyses.
If we extend the ordering of $\Powerset(D)$ pointwise to the monotone functions
$ \Powerset(D) \rightarrow \Powerset(D) $, the least element will be a function
that always returns $\emptyset$, i.e., $ f_\bot = \lambda l . \emptyset $.
Clearly we have that
\[
  \forall f : f_\bot \circ f = f_\bot
\]
but in the second case
\[
  \neg(\forall f : f \circ f_\bot = f_\bot)
\]
Therefore, such analyses do not directly fit in the in the original framework of
WPDS or CPDS\@.
Yet they do in our modified one that relaxes the requirement of annihilation.

%% file: conclusions.tex
Weighted/communicating pushdown systems have been used in many contexts and are
a popular approach to interprocedural analysis.
However, their requirements with respect to the abstract domain were quite
restrictive and did not admit some of the classical analyses directly.
In this paper we have shown that some of the restrictions are not necessary.
We have achieved that by reformulating the algorithms for backward and forward
reachability.
Furthermore, we have proved that they are sound --- they always provide a safe
over-approximation of the join over all valid paths solution.
Provided some additional properties of the abstract domain, we have also shown
that those solutions coincide, i.e., the algorithms are complete.

We believe that our results strengthen the connection between the monotone
frameworks and the pushdown systems by making it possible to directly express
more analyses based on monotone frameworks in the setting of pushdown systems.
Moreover, the development does provide some additional flexibility when both
designing and implementing analyses using pushdown systems.
For instance, the annihilation property might be useful for certain analyses,
but now this is the choice of the designer of the analysis and not a hard
requirement from the framework.
Last, but not least, we believe that the paper improves the understanding of
using weighted pushdown systems for interprocedural program analysis.

%% file: soundness-proofs.tex
\subsection{Proof of Thm.~\ref{thm:ann/pre-soundness}}
\label{app:ann/pre-soundness}

Consider an automaton $\Automaton$ and its corresponding $\APreStarC$ generated
by the saturation procedure. Let us assume that we have the least solution
$\PreSolSymb$ to the set of constraints $\Constraints$. Then for each pair $(p, s)$
such that $\Ang{p, s} \PdsTransS{\sigma} \Ang{q_f, \epsilon}$ (where $\sigma \in
\PdsRulesPre^{*}$ and $q_f \in F$), we have
$v(\sigma) \Leq \PreSolSymbP(\rho)$ where
$\rho = p \AutoTransS{s} q_f$ is in $\APreStar$.

\begin{proof}
Note that we do not need to prove the existence of the paths in the $\APreStar$
--- it is a previously known result~\cite{bib:schwoon-thesis,bib:wpds1}.
We can use it because our algorithm differs only in the constraint generation,
and not in the way new transitions are added.
Moreover, as explained above, the additional rules in $\PdsRulesPre$ do not
change that result.

The proof will proceed by induction on $|\sigma|$ (note that since $P$ and $F$
are disjoint, it is not possible to have $|\sigma| = 0$).
\begin{description}
\item[$|\sigma| = 1$]
  We know that the path in the pushdown system is
  $\Ang{p, \gamma} \PdsTrans{r} \Ang{q_f, \epsilon}$.
  But this means that $r \in \PdsRulesPre \setminus \PdsRules$.
  Existence of $p \AutoTrans{\gamma} q_f$ follows directly from the definition of
  $\PdsRulesPre$. We also have that $f(r) = \One$.
  Finally, according to the saturation procedure there exists a constraint:
  $\One \Leq \PreCon{p}{\gamma}{q_f}$.
  Therefore, clearly $v([r]) \Leq \PreSol{p}{\gamma}{q}$.

\item[$|\sigma| > 1$]
  In this case we know that the path in the pushdown system is
  \[
  \Ang{p, \gamma s_0} \PdsTrans{r} \Ang{q', w s_0} \PdsTransS{\sigma'} \Ang{q_f, \epsilon}
  \]
  for some $q', \gamma,$ and $w$.
  Moreover, $r = \Ang{p, \gamma} \PdsRule{} \Ang{q', w}$ where $s = \gamma s_0$.

  If $q' \not \in P$ then
  $r \in \PdsRulesPre \setminus \PdsRules$ and $f(r) = \One$ ($r$ is one of the
  added rules to $\PdsRulesPre$).
  Furthermore, all the rules of $\sigma'$ must also be in
  $\PdsRulesPre \setminus \PdsRules$ and thus there must be a path
  $\rho = p \AutoTransS{s} q_f$ in $\APreStar$ (since it must also be in
  $\Automaton$).
  Therefore, $v(\sigma) = \One$ and for each transition $t$ on the path
  $\rho$ we have a constraint of the form $\One \Leq \PreConSymb(t)$, thus by
  monotonicity we have $v(\sigma) \Leq \PreSolSymbP(\rho)$.

  Otherwise $q' \in P$ and $r \in \PdsRules$, so we can use the induction
  hypothesis to get that
  \[
  v(\sigma') \Leq \PreSolP{q'}{w s_0}{\rho'}{q_f}
  \]
  where
  \[
  \begin{array}{cccccc}
  & \multicolumn{3}{c}{\overbrace{\hspace{5em}}^{\rho'_1}} & & \\
  \rho' = & q' & \AutoTransS{w} & q'' & \AutoTransS{s_0} & q_f \\
  & & & \multicolumn{3}{c}{\underbrace{\hspace{5em}}_{\rho'_2}}
  \end{array}
  \]
  Now the saturation procedure must have added the transition $p \AutoTrans{\gamma}
  q''$. So we have a path $\rho = p \AutoTrans{\gamma} q'' \AutoTransS{s_0} q_f$
  along with a constraint:
  \begin{enumerate}
  \item \label{lem:ann/safety-case1}
    if $w = \epsilon$ (so $q' = q''$) the added constraint is
    \[
    f(r) \Leq \PreCon{p}{\gamma}{q'}
    \]
  \item \label{lem:ann/safety-case2}
    if $w = \gamma'$ the added constraint is
    \[
    f(r) \Extend \PreCon{q'}{\gamma'}{q''} \Leq \PreCon{p}{\gamma}{q''}
    \]
  \item \label{lem:ann/safety-case3}
    if $w = \gamma_1' \gamma_2'$ the added constraint is
    \[
    f(r) \Extend \PreCon{q'}{\gamma_1'}{q_x}
          \Extend \PreCon{q_x}{\gamma_2'}{q''} \Leq \PreCon{p}{\gamma}{q''}
    \]
  \end{enumerate}
  For case~\ref{lem:ann/safety-case1} we have:
  \begin{align*}
  v(\sigma) & =    f(r) \Extend v(\sigma') \\
            & \Leq f(r) \Extend \PreSolP{q''}{s_0}{\rho_2'}{q_f} \\
            & \Leq \PreSol{p}{\gamma}{q'} \Extend \PreSolP{q''}{s_0}{\rho_2'}{q_f} \\
            & =    \PreSolP{p}{s}{\rho}{q_f} \\
  \end{align*}
  And for both~\ref{lem:ann/safety-case2} and~\ref{lem:ann/safety-case3}:
  \begin{align*}
  v(\sigma) & =    f(r) \Extend v(\sigma') \\
            & \Leq f(r) \Extend \PreSolP{q'}{w}{\rho_1'}{q''}
                        \Extend \PreSolP{q''}{s_0}{\rho_2'}{q_f} \\
            & \Leq \PreSolP{p}{\gamma}{q'}
                \Extend \PreSolP{q'}{w}{\rho_1'}{q''}
                \Extend \PreSolP{q''}{s_0}{\rho_2'}{q_f} \\
            & =    \PreSolP{p}{s}{\rho}{q_f} \\
  \end{align*}

  Thus in all possible cases we have that:
  \[
    v(\sigma) \Leq \PreSolP{p}{s}{\rho}{q_f}
  \]
\end{description}
\qed
\end{proof}

\subsection{Proof of Thm.~\ref{thm:ann/post-soundness}}
\label{app:ann/post-soundness}

Consider an automaton $\Automaton$ and its corresponding $\APostStarC$ generated
by the saturation procedure. Let us assume that we have the least solution
$\PostSolSymb$ to the set of constraints $\Constraints$. Then for each pair $(p, s)$
such that $\Ang{q_f, \epsilon} \PdsTransS{\sigma} \Ang{p, s}$ (where $\sigma \in
\PdsRulesPost^{*}$ and $q_f \in F$), we have
$v(\sigma) \Leq \PostSolSymbP(\rho)$ where
$\rho = q_f \AutoTransRS{s} p$ is in $\APostStarC$.

\begin{proof}
Note that, as in the case of $\PreStar$, we do not need to prove the existence
of the paths in the $\APreStar$ --- it is a previously known
result~\cite{bib:schwoon-thesis,bib:wpds1}.
Again this is due to the fact that our algorithm differs only in the constraint
generation, and not in the way new transitions are added.
Moreover, as explained above, the additional rules in $\PdsRulesPost$ do not
change that result.

The proof will proceed by induction on $|\sigma|$ (note that since $P$ and $F$
are disjoint, it is not possible to have $|\sigma| = 0$).
\begin{description}
\item[$|\sigma| = 1$]
  So $s = \gamma$ and we have $\Ang{q_f, \epsilon} \PdsTrans{r} \Ang{p, s}$. We
  know that $r \in \PdsRulesPost \setminus \PdsRules$, and so from the
  definition of $\PdsRulesPost$ we have that there is transition
  $q_f \AutoTransR{\gamma} p$ and $v([r]) = \One$. Moreover, from the saturation
  procedure we have a constraint $\One \Leq \PostCon{p}{\gamma}{q_f}$.
  Therefore, $v([r]) \Leq \PostSol{p}{\gamma}{q_f}$.

\item[$|\sigma| > 1$]
  So we have
  \[
  \Ang{q_f, \epsilon} \PdsTransS{\sigma'} \Ang{q', s'} \PdsTrans{r} \Ang{p, s}
  \]
  where $\sigma = \sigma' r$.

  If $q' \not \in P$ then
  $r \in \PdsRulesPost \setminus \PdsRules$ and
  it must be of the form
  $r = \Ang{q', \epsilon} \PdsRule{} \Ang{p, \gamma}$ where $s = \gamma s'$ ($r$
  is one of the additional rules to the $\PdsRulesPost$).
  But that means that all the remaining rules in $\sigma'$ must also be one of
  those additional rules ($\PdsRulesPost \setminus \PdsRules$).
  Thus the weight of every transition $t$ on the path $q_f \AutoTransRS{s} p$ is
  $\PostSolSymb(t) = \One$ (its existence follows directly from the definition
  of $\PdsRulesPost$).
  Moreover, all of them must have a corresponding constraint of the form
  $\One \Leq \PostConSymb(t)$.
  Therefore, by monotonicity we have $\One \Leq \PostSolP{p}{s}{\rho}{q_f}$
  and so $v(\sigma) \Leq \PostSolP{p}{s}{\rho}{q_f}$.

  Otherwise $q' \in P$ and $r \in \PdsRules$,
  $r = \Ang{q', \gamma'} \PdsRule{} \Ang{p, w}$ and $s = w s_0$,
  $s' = \gamma' s_0$.
  Since $|\sigma'| < |\sigma|$ we can use the induction hypothesis to get that
  \[
  v(\sigma') \Leq \PostSolP{q'}{s'}{\rho'}{q_f}
  \]
  where
  \[
  \begin{array}{cccccc}
  & \multicolumn{3}{c}{\overbrace{\hspace{5em}}^{\rho'_2}} & & \\
  \rho' = q_f \AutoTransRS{s'} q' = & q_f & \AutoTransRS{s_0} & q'' & \AutoTransRE{\gamma'} & q' \\
  & & & \multicolumn{3}{c}{\underbrace{\hspace{5em}}_{\rho'_1}}
  \end{array}
  \]
  for some $q''$.
  And so we have three possibilities, depending on $w$:
  \begin{enumerate}
  \item if $w = \epsilon$, the transition $q'' \AutoTransR{\epsilon} p$ along with
    the following constraint
    \[
    \PostConPE{q'}{\gamma'}{\rho_1'}{q''}
      \Extend f(r) \Leq \PostCon{p}{\epsilon}{q''}
    \]
    Therefore, the solution will have to satisfy:
    \[
    \PostSolP{q'}{\gamma'}{\rho_1'}{q''}
      \Extend f(r) \Leq \PostSol{p}{\epsilon}{q''}
    \]
    and so
    \begin{align*}
    v(\sigma)
      & =       v(\sigma') \Extend f(r) \\
      & \Leq    \PostSolP{q'}{s'}{\rho'}{q_f} \Extend f(r) \\
      & =       \PostSolP{q''}{s_0}{\rho_2'}{q_f}
        \Extend \PostSolP{q'}{\gamma'}{\rho_1'}{q''}
        \Extend f(r) \\
      & \Leq    \PostSolP{q''}{s_0}{\rho_2'}{q_f}
        \Extend \PostSol{p}{\epsilon}{q''} \\
      & =       \PostSolP{p}{s}{\rho}{q_f}
    \end{align*}
  \item if $w = \gamma$, the transition $q'' \AutoTransR{\gamma} p$ along with
    the following constraint
    \[
    \PostConPE{q'}{\gamma'}{\rho_1'}{q''}
      \Extend f(r) \Leq \PostCon{p}{\gamma}{q''}
    \]
    Therefore, the solution will have to satisfy:
    \[
    \PostSolP{q'}{\gamma'}{\rho_1'}{q''}
      \Extend f(r) \Leq \PostSol{p}{\gamma}{q''}
    \]
    and so
    \begin{align*}
    v(\sigma)
      & =       v(\sigma') \Extend f(r) \\
      & \Leq    \PostSolP{q'}{s'}{\rho'}{q_f} \Extend f(r) \\
      & =       \PostSolP{q''}{s_0}{\rho_2'}{q_f}
        \Extend \PostSolP{q'}{\gamma'}{\rho_1'}{q''}
        \Extend f(r) \\
      & \Leq    \PostSolP{q''}{s_0}{\rho_2'}{q_f}
        \Extend \PostSol{p}{\gamma}{q''} \\
      & =       \PostSolP{p}{s}{\rho}{q_f}
    \end{align*}
  \item if $w = \gamma_1 \gamma_2$, the transitions
    $q_{p, \gamma_1} \AutoTransR{\gamma_1} q'$ and
    $q'' \AutoTransR{\gamma_2} q_{p, \gamma_1}$ along with
    the following constraints
    \[
    \One \Leq \PostCon{q}{\gamma_1}{q_{p, \gamma_1}}
    \]
    and
    \[
    \PostConPE{q'}{\gamma'}{\rho_1'}{q''}
      \Extend f(r) \Leq \PostCon{q_{p, \gamma_1}}{\gamma_2}{q''}
    \]
    Therefore, the solution will have to satisfy:
    \[
    \One \Leq \PostSol{q'}{\gamma_1}{q_{p, \gamma_1}}
    \]
    \[
    \PostSolP{q'}{\gamma'}{\rho_1'}{q''}
      \Extend f(r) \Leq \PostSol{q_{p, \gamma_1}}{\gamma_2}{q''}
    \]
    and so
    \begin{align*}
    v(\sigma)
      & =       v(\sigma') \Extend f(r) \\
      & \Leq    \PostSolP{q'}{s'}{\rho'}{q_f} \Extend f(r) \\
      & =       \PostSolP{q''}{s_0}{\rho_2'}{q_f}
        \Extend \PostSolP{q'}{\gamma'}{\rho_1'}{q''}
        \Extend f(r) \\
      & \Leq    \PostSolP{q''}{s_0}{\rho_2'}{q_f}
        \Extend \PostSol{p}{\gamma}{q''} \\
      & =       \PostSolP{q''}{s_0}{\rho_2'}{q_f}
        \Extend \PostSol{q_{p, \gamma_1}}{\gamma_2}{q''}
        \Extend \PostSol{q'}{\gamma_1}{q_{p, \gamma_1}} \\
      & =       \PostSolP{p}{s}{\rho}{q_f}
    \end{align*}
  \end{enumerate}
\end{description}
\qed
\end{proof}

%% file: continuity-proof.tex
The function $F$, defined as:
\begin{align*}
  & F : (\Transitions \rightarrow D) \rightarrow (\Transitions \rightarrow D)
  \\
  & F(m) t = \BigCombine_{c \in \Constraints_t} \Lhs_m(c)
\end{align*}
is continuous, i.e, for any non-empty chain $Y$:
  \[
    F(\BigLub Y) = \BigLub_{m \in Y} F(m)
  \]
\begin{proof}
  Since we are assuming that $D$ is a complete lattice and $m$ is a total
  function, then $ \Transitions \rightarrow D $ defines a complete lattice as
  well.
  Furthermore, we have that for any $ Y \subseteq \Transitions \rightarrow D $
  \begin{equation}
    \label{eq:functional-lub-property}
    (\BigLub Y) t = \BigCombine_{m \in Y} m(t)
  \end{equation}

  Therefore, we have:

  \begin{align*}
    & F(\BigLub Y) t
    \\
    & \qquad = \ProofEx{definition of $F$}
    \\
    & \BigCombine \Set{ lhs_{\BigLub Y}(c) \mid c \in \Constraints_t }
    \\
    & \qquad = \ProofEx{equation \eqref{eq:functional-lub-property}}
    \\
    & \BigCombine \Set{ lhs_{\lambda t' . \BigCombine_{m \in Y} m(t')}(c)
                   \mid c \in \Constraints_t }
    \\
    & \qquad = \ProofEx{$D$ is affine, $Y$ is not empty and
                        the constraints are finite}
    \\
    & \BigCombine \Set{ \BigCombine_{m \in Y} lhs_m(c) \mid c \in \Constraints_t }
    \\
    & \qquad = \ProofEx{$D$ is a complete lattice}
    \\
    & \BigCombine_{m \in Y} (\BigCombine \Set{ lhs_m(c) \mid c \in \Constraints_t })
    \\
    & \qquad = \ProofEx{definition of $F$}
    \\
    & \BigCombine_{m \in Y} F(m) t
    \\
    & \qquad = \ProofEx{equation \eqref{eq:functional-lub-property}}
    \\
    & (\BigLub_{m \in Y} F(m)) t
  \end{align*}
  \qed
\end{proof}

%% file: completeness-proofs.tex
\subsection{Proof of Lem.~\ref{lem:ann/pre-exists-path}}
\label{app:ann/pre-exists-path}

For every transition $q \AutoTrans{\gamma} q'$ in $\APreStar$ there exists a
sequence $\sigma \in \PdsRulesPre$ such that
$\Ang{q, \gamma} \PdsTransS{\sigma} \Ang{q', \epsilon}$.

\begin{proof}
Proof will proceed by induction on $\Automaton_i$, where $\Automaton_i$
corresponds to the initial automaton after $i$ steps of the saturation
procedure.
\begin{description}
\item[$i = 0$]
  Follows from the definition of $\PdsRulesPre$.
\item[$i > 0$]
  We assume the property holds for $\Automaton_i$ and prove it for
  $\Automaton_{i+1}$.
  Consider that the saturation procedure adds a transition $p_s \AutoTrans{\gamma} q_d$
  (note that the saturation procedure works on $\PdsRules$) because of:
  \begin{itemize}
  \item a pushdown rule $r = \Ang{p_s, \gamma} \PdsRule{} \Ang{q_d, \epsilon}$.
    The result is immediate from the rule.
  \item a pushdown rule $r = \Ang{p_s, \gamma} \PdsRule{} \Ang{p', \gamma'}$
    and a transition $p' \AutoTrans{\gamma'} q_d$ in $\Automaton_i$.
    We use the induction hypothesis on $p' \AutoTrans{\gamma'} q_d$ and get that
    there exists $\sigma$ such that
    $\Ang{p', \gamma'} \PdsTransS{\sigma} \Ang{q_d, \epsilon}$.
    But then we also have that
    \[
    \Ang{p_s, \gamma} \PdsTrans{r} \Ang{p', \gamma'} \PdsTransS{\sigma} \Ang{q_d, \epsilon}
    \]
  \item a pushdown rule $r = \Ang{p_s, \gamma} \PdsRule{} \Ang{p', \gamma' \gamma''}$
    and a path $p' \AutoTrans{\gamma'} q'' \AutoTrans{\gamma''} q_d$ in $\Automaton_i$.
    We use the induction hypothesis on $p' \AutoTrans{\gamma'} q''$ and
    $q'' \AutoTrans{\gamma''} q_d$
    to get that there exists $\sigma'$ and $\sigma''$ such that
    $\Ang{p', \gamma'} \PdsTransS{\sigma'} \Ang{q'', \epsilon}$
    and
    $\Ang{q'', \gamma'} \PdsTransS{\sigma''} \Ang{q_d, \epsilon}$.
    And again we have that:
    \[
    \Ang{p_s, \gamma} \PdsTrans{r} \Ang{q', \gamma' \gamma''}
                      \PdsTransS{\sigma' \sigma''} \Ang{q_d, \epsilon}
    \]
  \end{itemize}
\end{description}
\qed
\end{proof}

\subsection{Proof of Lem.~\ref{lem:ann/pre-completeness-transition}}
\label{app:ann/pre-completeness-transition}

Consider a weighted pushdown system $\WPds = (\Pds, \FlowAlgebra, f)$ where
$\FlowAlgebra$ is affine and an automaton $\APreStarC$ created by the saturation
procedure.
For every transition $q \AutoTrans{\gamma} q'$ in this automaton we have that
\[
\PreSol{q}{\gamma}{q'} \Leq \BigCombine \Set{ v(\sigma) \mid
  \Ang{q, \gamma} \PdsTransS{\sigma} \Ang{q', \epsilon}, \sigma \in \PdsRulesPre^{*} }
\]

\begin{proof}
Let us also denote by $\AutomatonC_i$ the automaton $\Automaton$ after $i$ steps
of the saturation procedure. Also let us denote the least solution for
$\AutomatonC_i$ by $\PostSolSymb_i$. We will prove by induction on $i$ that for
every transition $q \AutoTrans{\gamma} q'$ in $\AutomatonC_i$ we have that
\[
\PreSolC{_i}{q}{\gamma}{q'} \Leq \BigCombine \Set{ v(\sigma) \mid
  \Ang{q, \gamma} \PdsTransS{\sigma} \Ang{q', \epsilon}, \sigma \in \PdsRulesPre^{*} }
\]

\begin{description}
\item[$i = 0$] $\AutomatonC_0$ is just the initial automaton $\Automaton$ with
  the set $\Constraints$ containing one constraint for every transition of
  $\Automaton$. The property clearly holds.

\item[$i > 0$] We assume the property holds for $\AutomatonC_i$ and prove it for
  $\AutomatonC_{i+1}$, i.e., prove that adding a constraint (and maybe a
  transition as well) preserves the property of interest.

  Let $t$ be the transition that the added constraint refers to.
  Observe that if $t$ was already in the automaton $\AutomatonC_i$, then it is
  possible that $\PreSolSymb(t)$ might be on the left-hand side of some other
  constraint.
  Therefore, the least solution for the new set of constraints might be different
  for other transitions as well; in other words the value/information from the
  new constraint might have to be propagated throughout other constraints to get
  $\PostSolSymb_{i+1}$.
  Now let $\PostSolSymb_{i}^{j}$ denote the solution after $j$ steps of fixed
  point computation with the new constraint, starting with
  \[
  \PreSolSymb_{i}^0(t) =
    \begin{cases}
      \Zero              & \text{if $t$ was added} \\
      \PreSolSymb_{i}(t) & \text{otherwise ($t$ was in $\AutomatonC_i$)}
    \end{cases}
  \]
  Using induction on $j$ we will prove
  that the property of interest is maintained by the computation.

  Note that we can use here Kleene iteration due to Lemma~\ref{lem:continuous}.

  \begin{description}
  \item[$j=0$] Immediate from outer induction hypothesis.
  \item[$j>0$] In the following we will use the fact that the flow algebra is
    affine; it is enough for our purposes because from
    Lemma~\ref{lem:ann/pre-exists-path} it follows that the sets (of pushdown
    paths) on the right-hand sides are not empty. Let us consider each form of
    the possible constraints:
    \begin{itemize}
    \item $f(r) \Leq \PreSol{q}{\gamma}{q'}$ where
      $r = \Ang{q, \gamma} \PdsRule{} \Ang{q', \epsilon}$.
      We know that
      \[
      \PreSolC{_i^{j+1}}{q}{\gamma}{q'} = \PreSolC{_i^j}{q}{\gamma}{q'} \Combine f(r)
      \]
      Moreover, from the rule $r$ it immediately follows that
      \[
      f(r) \Leq \BigCombine \Set{ v(\sigma) \mid
        \Ang{q, \gamma} \PdsTransS{\sigma} \Ang{q', \epsilon} }
      \]
      Using this and the induction hypothesis on $\PreSolC{_i^j}{q}{\gamma}{q'}$
      \[
      \PreSolC{_i^{j+1}}{q}{\gamma}{q'} \Leq \BigCombine \Set{ v(\sigma) \mid
        \Ang{q, \gamma} \PdsTransS{\sigma} \Ang{q', \epsilon} }
      \]
    \item $f(r) \Extend \PreSol{q''}{\gamma''}{q'} \Leq \PreSol{q}{\gamma}{q'}$
      where $r = \Ang{q, \gamma} \PdsRule{} \Ang{q'', \gamma''}$ and
      $q'' \AutoTrans{\gamma''} q'$. We have that
      \[
      \PreSolC{_i^{j+1}}{q}{\gamma}{q'} = \PreSolC{_i^j}{q}{\gamma}{q'}
        \Combine ( f(r) \Extend \PreSolC{_i^j}{q''}{\gamma''}{q'} )
      \]
      Now let us use the induction hypothesis:
      \[
      \PreSolC{_i^j}{q''}{\gamma''}{q'} \Leq \BigCombine \Set{ v(\sigma) \mid
        \Ang{q'', \gamma''} \PdsTransS{\sigma} \Ang{q', \epsilon} }
      \]
      Multiplying both sides by $f(r)$ and using that $\Extend$ is
      affine:
      \begin{align*}
      f(r) \Extend \PreSolC{_i^j}{q''}{\gamma''}{q'}
        & \Leq \BigCombine \Set{ f(r) \Extend v(\sigma) \mid
          \Ang{q'', \gamma''} \PdsTransS{\sigma} \Ang{q', \epsilon} }
          \\
        & \Leq \BigCombine \Set{ v(\sigma) \mid
          \Ang{q, \gamma} \PdsTransS{\sigma} \Ang{q', \epsilon} }
      \end{align*}
      Therefore:
      \[
      \PreSolC{_i^{j+1}}{q}{\gamma}{q'} \Leq \BigCombine \Set{ v(\sigma) \mid
        \Ang{q, \gamma} \PdsTransS{\sigma} \Ang{q', \epsilon} }
      \]
    \item $f(r) \Extend \PreSol{q''}{\gamma_1''}{q_1'}
                \Extend \PreSol{q_1'}{\gamma_2''}{q'}
                \Leq    \PreSol{q}{\gamma}{q'}$
      where
      $r = \Ang{q, \gamma} \PdsRule{} \Ang{q'', \gamma_1'' \gamma_2''}$ and
      $q'' \AutoTrans{\gamma_1''} q_1' \AutoTrans{\gamma_2''} q'$.
      We have that
      \begin{align*}
      \PreSolC{_i^{j+1}}{q}{\gamma}{q'}
        & = \PreSolC{_i^j}{q}{\gamma}{q'} \\
        & \Combine ( f(r) \Extend \PreSolC{_i^j}{q''}{\gamma_1''}{q_1'}
                          \Extend \PreSolC{_i^j}{q_1'}{\gamma_2''}{q'} )
      \end{align*}
      We use the induction hypothesis twice to get
      \begin{align*}
      \PreSolC{_i^j}{q''}{\gamma_1''}{q_1'}
        & \Leq \BigCombine \Set{ v(\sigma) \mid
          \Ang{q'', \gamma_1''} \PdsTransS{\sigma} \Ang{q_1', \epsilon} }
        \\
      \PreSolC{_i^j}{q_1'}{\gamma_2''}{q'}
        & \Leq \BigCombine \Set{ v(\sigma) \mid
          \Ang{q_1', \gamma_2''} \PdsTransS{\sigma} \Ang{q', \epsilon} }
        \\
      \end{align*}
      From monotonicity and the fact that $\Extend$ is affine we get that:
      \begin{align*}
      & f(r) \Extend \PreSolC{_i^j}{q''}{\gamma_1''}{q_1'}
             \Extend \PreSolC{_i^j}{q_1'}{\gamma_2''}{q'}
        \\
      & \qquad \Leq
        \BigCombine \Set{ f(r) \Extend v(\sigma_1) \Extend v(\sigma_2) \mid
          \Ang{q'', \gamma_1''} \PdsTransS{\sigma_1} \Ang{q_1', \epsilon},
          \Ang{q_1', \gamma_2''} \PdsTransS{\sigma_2} \Ang{q', \epsilon} }
        \\
      & \qquad \Leq
        \BigCombine \Set{ f(r) \Extend v(\sigma) \mid
          \Ang{q'', \gamma_1'' \gamma_2''} \PdsTransS{\sigma} \Ang{q', \epsilon} }
        \\
      & \qquad \Leq
        \BigCombine \Set{ v(\sigma) \mid
          \Ang{q, \gamma} \PdsTransS{\sigma} \Ang{q', \epsilon} }
      \end{align*}
      Therefore
      \[
      \PreSolC{_i^{j+1}}{q}{\gamma}{q'} \Leq \BigCombine \Set{ v(\sigma) \mid
        \Ang{q, \gamma} \PdsTransS{\sigma} \Ang{q', \epsilon} }
      \]
    \end{itemize}
  \end{description}

\end{description}
\qed
\end{proof}

\subsection{Proof of Lem.~\ref{lem:ann/pre-completeness-path}}
\label{app:ann/pre-completeness-path}

Consider a weighted pushdown system $\WPds = (\Pds, \FlowAlgebra, f)$ where
$\FlowAlgebra$ is affine and a $\APreStarC$ automaton created by the saturation
procedure. For every path $\rho = q \AutoTransS{s} q'$ in this automaton we have
that
\[
\PreSolP{q}{s}{\rho}{q'} \Leq \BigCombine \Set{ v(\sigma) \mid
  \Ang{q, s} \PdsTransS{\sigma} \Ang{q', \epsilon}, \sigma \in \PdsRulesPre^{*} }
\]

\begin{proof}
The proof will proceed with the induction on the number of transitions
$|\rho|$ (we will use the inductive definition of $\PreSolSymbP$).
\begin{description}
\item[$|\rho| = 1$] So $\rho$ is just a single transition, therefore according
  to the definition of $\PreSolSymb$ we have
  \[
  \PreSolP{q}{s}{\rho}{q'} = \PreSol{q}{s}{q'}
  \]
  The result follows from Lemma~\ref{lem:ann/pre-completeness-transition}.
\item[$1 < |\rho|$] Again using the definition of $\PreSolSymb$ we have
  \[
  \PreSolP{q}{s}{\rho}{q'} =
    \PreSol{q}{\gamma}{q''} \Extend \PreSolP{q''}{s'}{\rho'}{q'}
  \]
  where $s = \gamma s'$, $q'' \in Q$, and
  \[
  \rho = q \AutoTrans{\gamma} \underbrace{q'' \AutoTransS{s'} q'}_{\rho'}
  \]
  Now we can use Lemma~\ref{lem:ann/pre-completeness-transition} again and
  the induction hypothesis (since $|\rho|' < |\rho|$) to get:
  \begin{align*}
  \PreSol{q}{\gamma}{q''}
    & \Leq
    \BigCombine \Set{ v(\sigma) \mid
      \Ang{q, \gamma} \PdsTransS{\sigma} \Ang{q'', \epsilon}, \sigma \in \PdsRulesPre
    } \\
  \PreSolP{q''}{s'}{\rho'}{q'}
    & \Leq
    \BigCombine \Set{ v(\sigma) \mid
      \Ang{q'', s'} \PdsTransS{\sigma} \Ang{q', \epsilon}, \sigma \in \PdsRulesPre
    }
  \end{align*}
  Finally, we use the fact that the flow algebra is affine:
  \begin{align*}
  & \PreSolP{q}{s}{\rho}{q'} \\
    & \Leq
    \BigCombine \Set{ v(\sigma) \Extend v(\sigma') \mid
      \Ang{q, \gamma} \PdsTransS{\sigma} \Ang{q'', \epsilon},
      \Ang{q'', s'} \PdsTransS{\sigma'} \Ang{q', \epsilon},
      \sigma, \sigma' \in \PdsRulesPre
    } \\
    & \Leq
    \BigCombine \Set{ v(\sigma) \mid
      \Ang{q, s} \PdsTransS{\sigma} \Ang{q', \epsilon}, \sigma \in \PdsRulesPre
    }
  \end{align*}
\end{description}
\qed
\end{proof}

\subsection{Proof of Thm.~\ref{thm:ann/pre-completeness}}
\label{app:ann/pre-completeness}

Consider an automaton $\APreStarC$ constructed by the saturation procedure and
the least solution $\PreSolSymb$ to the set of its constraints $\Constraints$.
If the flow algebra is affine then for every path
$\rho = p \AutoTransS{s} q_f$ where $q_f \in F$ we have that
\[
  \PreSolP{p}{s}{\rho}{q_f} = \BigCombine
  \Set{ v(\sigma) \mid \Ang{p, s} \PdsTransS{\sigma} \Ang{q_f, \epsilon},
                       \sigma \in \PdsRulesPre^{*}
  }
\]

\begin{proof}
The result follows directly from Theorem~\ref{thm:ann/pre-soundness}
and Lemma~\ref{lem:ann/pre-completeness-path}.
\qed
\end{proof}

%
%

\subsection{Proof of Lem.~\ref{lem:ann/post-exists-path}}
\label{app:ann/post-exists-path}

For every transition $q' \AutoTransR{\gamma_\epsilon} q$
($\gamma_\epsilon \in \Gamma \cup \{ \epsilon \}$)
in $\APostStar$ there exists a sequence $\sigma$ of pushdown rules in
$\PdsRulesPostP$ such that
$\Ang{q', \epsilon} \PdsTransS{\sigma} \Ang{q, \gamma_\epsilon}$.

\begin{proof}
Let us denote by $\Automaton_i$ the automaton $\Automaton$ after $i$ steps of
the saturation procedure. Proof will proceed by induction on $i$.
\begin{description}
\item[$i = 0$]
  Follows from the definition of $\PdsRulesPostP$.
\item[$i > 0$]
  We assume the property holds for $\Automaton_i$ and prove it for
  $\Automaton_{i+1}$.
  Consider that the saturation procedure\footnote{Note that the saturation
  procedure works on $\PdsRules$.}
  \begin{itemize}
  \item adds a transition $q_d \AutoTransR{\epsilon} p_s$ because of a pushdown rule
    $r = \Ang{p', \gamma'} \PdsRule{} \Ang{p_s, \epsilon}$ and a path
    $q_d \AutoTransRE{\gamma'} p'$. We can use the induction hypothesis to get
    that there exists $\sigma$ such that
    $\Ang{q_d, \epsilon} \PdsTransS{\sigma} \Ang{p', \gamma'}$. But then clearly
    $\Ang{q_d, \epsilon} \PdsTransS{\sigma} \Ang{p', \gamma'}
                         \PdsTrans{r} \Ang{p, \epsilon}$.
  \item adds a transition $q_d \AutoTransR{\gamma} p_s$ because of a pushdown rule
    $r = \Ang{p', \gamma'} \PdsRule{} \Ang{p_s, \epsilon}$ and a path
    $q_d \AutoTransRE{\gamma'} p'$. We can use the induction hypothesis to get
    that there exists $\sigma$ such that
    $\Ang{q_d, \epsilon} \PdsTransS{\sigma} \Ang{p', \gamma'}$.
    Again it is clear that
    $\Ang{q_d, \epsilon} \PdsTransS{\sigma} \Ang{p', \gamma'}
                         \PdsTrans{r} \Ang{p, \epsilon}$.
  \item adds $q_{p_s, \gamma_1} \AutoTransR{\gamma_1} p_s$ and
    $q_d \AutoTransR{\gamma_2} q_{p_s, \gamma_1}$
    because of a pushdown rule
    $r = \Ang{p', \gamma'} \PdsRule{} \Ang{p_s, \gamma_1 \gamma_2 }$ and a
    path $q_d \AutoTransRE{\gamma'} p'$. According to the definition of
    $\PdsRulesPostP$ we know that there are
    $r_1 = \Ang{p', \gamma'} \PdsRule{} \Ang{q_{p_s, \gamma_1}, \gamma_2}$
    and
    $r_2 = \Ang{q_{p_s, \gamma_1}, \epsilon} \PdsRule{} \Ang{p_s, \gamma_1}$.
    So we immediately have the path for the first transition:
    \[
    \Ang{q_{p_s, \gamma_1}, \epsilon} \PdsTrans{r_2} \Ang{p_s, \gamma_1}
    \]
    Moreover, we can use the induction hypothesis to get that there exists
    $\sigma$ such that
    $\Ang{q_d, \epsilon} \PdsTransS{\sigma} \Ang{p', \gamma'}$
    and so we also have that
    \[
    \Ang{q_d, \epsilon} \PdsTransS{\sigma} \Ang{p', \gamma'}
                         \PdsTrans{r_1} \Ang{q_{p_s, \gamma_1}, \gamma_2}
    \]
  \end{itemize}
\end{description}
\qed
\end{proof}

\subsection{Proof of Lem.~\ref{lem:ann/post-completeness-transition}}
\label{app:ann/post-completeness-transition}

Consider a weighted pushdown system $\WPds = (\Pds, \FlowAlgebra, f)$ where
$\FlowAlgebra$ is affine and an automaton $\APostStarC$ created by the
saturation procedure.
For every transition $q' \AutoTransR{\gamma_\epsilon} q$
($\gamma_\epsilon \in \Gamma \cup \{ \epsilon \}$)
in this automaton we have that
\[
\PostSol{q}{\gamma_\epsilon}{q'} \Leq \BigCombine \Set{ v(\sigma) \mid
  \Ang{q', \epsilon} \PdsTransS{\sigma} \Ang{q, \gamma_\epsilon}, \sigma \in
  \PdsRulesPostP^{*} }
\]

\begin{proof}
Let us denote by $\AutomatonC_i$ the automaton $\AutomatonC$ after $i$ steps
of saturation procedure and similarly the least solution for
it by $\PostSolSymb_i$.
We will prove by induction on $i$ that for every transition $q'
\AutoTransR{\gamma_\epsilon} q$ we have that
\[
\PostSolC{_i}{q}{\gamma_\epsilon}{q'} \Leq \BigCombine \Set{ v(\sigma) \mid
  \Ang{q', \epsilon} \PdsTransS{\sigma} \Ang{q, \gamma_\epsilon}, \sigma \in
  \PdsRulesPostP^{*} }
\]

\begin{description}
\item[$i = 0$]
  The only constraints are of the form $\One \Leq l(t)$ where $t$ is a
  transition in $\Automaton$.
  Therefore, the least solution for each $t$ is $\PostSolSymb_i(t) = \One$.
  We also know that for every
  $r \in \PdsRulesPostP \setminus \PdsRules$, $f(r) = \One$.
  So the right hand side is at least $\One$.
  Thus our property holds.

\item[$i > 0$]
  We assume the property holds for $\AutomatonC_i$ and prove it for
  $\AutomatonC_{i+1}$, i.e., prove that adding a constraint (and maybe a
  transition as well) preserves the property of interest.

  Let $t$ bi the transition that the added constraint refers to.
  Observe that if $t$ was already in the automaton $\AutomatonC_i$, then it is
  possible that $\PostConSymb(t)$ might be on the left-hand side of some other
  constraint.
  Therefore, the least solution for the new set of constraints might be different
  for other transitions as well; in other words the value/information from the
  new constraint might have to be propagated throughout other constraints
  to get $\PostSolSymb_{i+1}$.
  Now let $\PostSolSymb_{i}^{j}$ denote the solution after $j$ steps of fixed
  point computation with the new constraint, starting with
  \[
  \PostSolSymb_{i}^0(t) =
    \begin{cases}
      \Zero               & \mbox{if $t$ was added} \\
      \PostSolSymb_{i}(t) & \mbox{otherwise ($t$ was in $\AutomatonC_i$)}
    \end{cases}
  \]
  Using induction on $j$ we will prove that the property is maintained by the
  computation.

  Note that we can use here Kleene iteration due to Lemma~\ref{lem:continuous}.

  \begin{description}
  \item[$j=0$] Immediate from outer induction hypothesis.
  \item[$j>0$]
    We assume the property hold for $\PostSolSymb_i^j$ and prove that it also
    holds for $\PostSolSymb_i^{j+1}$.
    In the following we use the fact that the flow algebra is affine, this is
    enough since from Lemma~\ref{lem:ann/post-exists-path} it follows that the
    sets (of pushdown paths) on the right hand sides are not empty.
    Let us consider three possibilities of constraints:
    \begin{itemize}
    \item if the constraint is
      \[
      \PostCon{p'}{\gamma'}{q} \Extend f(r) \Leq \PostCon{p}{\epsilon}{q}
      \]
      or
      \[
      \PostCon{q''}{\gamma'}{q} \Extend \PostCon{p'}{\epsilon}{q''}
                           \Extend f(r) \Leq \PostCon{p}{\epsilon}{q}
      \]
      where
      $r = \Ang{p', \gamma'} \PdsRule{} \Ang{p, \epsilon} \in \PdsRules$.
      Let us only consider the more complex case with additional $\epsilon$
      transition (the one without is similar).
      We need to calculate the value of $\PostSolC{_i^{j+1}}{p}{\epsilon}{q}$
      --- it should be its old value combined with the new one
      \[
      \PostSolC{_i^{j+1}}{p}{\epsilon}{q} = \PostSolC{_i^j}{p}{\epsilon}{q}
                                   \Combine \left(
                                            \PostSolC{_i^j}{q''}{\gamma'}{q}
                                            \Extend \PostSolC{_i^j}{p'}{\epsilon}{q''}
                                            \Extend f(r)
                                            \right)
      \]
      Let us use the induction hypothesis (inner induction) three times to get:
      \begin{align*}
      \PostSolC{_i^j}{p}{\epsilon}{q} & \Leq \BigCombine \Set{ v(\sigma) \mid
          \Ang{q, \epsilon} \PdsTransS{\sigma} \Ang{p, \epsilon}, \sigma \in
          \PdsRulesPostP^{*} }
          \\
      \PostSolC{_i^j}{q''}{\gamma'}{q} & \Leq \BigCombine \Set{ v(\sigma) \mid
          \Ang{q, \epsilon} \PdsTransS{\sigma} \Ang{q'', \gamma'}, \sigma \in
          \PdsRulesPostP^{*} }
          \\
      \PostSolC{_i^j}{p'}{\epsilon}{q''} & \Leq \BigCombine \Set{ v(\sigma) \mid
          \Ang{q'', \epsilon} \PdsTransS{\sigma} \Ang{p', \epsilon}, \sigma \in
          \PdsRulesPostP^{*} }
          \\
      \end{align*}
      Using the above and the fact that our flow algebra is affine, we get:
      \begin{align*}
      & \PostSolC{_i^j}{q''}{\gamma'}{q} \Extend \PostSolC{_i^j}{p'}{\epsilon}{q''}
                                         \Extend f(r)  \\
      & \quad \Leq \BigCombine
          \Set{ v(\sigma_1) \Extend v(\sigma_2) \Extend f(r) \mid
              \Ang{q, \epsilon}  \PdsTransS{\sigma_1} \Ang{q'', \gamma'},
            \\
            & \hspace{50mm} \Ang{q'', \epsilon} \PdsTransS{\sigma_2} \Ang{p', \epsilon},
            \\
            & \hspace{50mm} \Ang{p', \gamma'} \PdsTrans{r} \Ang{p, \epsilon}, \\
            & \hspace{50mm}           \sigma \in \PdsRulesPostP^{*} }
         \\
      & \quad \Leq \BigCombine
          \Set{ v(\sigma) \mid
            \Ang{q, \epsilon} \PdsTransS{\sigma} \Ang{p, \epsilon},
            \sigma \in \PdsRulesPostP^{*} }
      \end{align*}
      Now since $\BigCombine$ gives the least upper bound, we have that
      \[
      \PostSolC{_i^{j+1}}{p}{\epsilon}{q}
        \Leq \BigCombine \Set{ v(\sigma) \mid
            \Ang{q, \epsilon} \PdsTransS{\sigma} \Ang{p, \epsilon},
            \sigma \in \PdsRulesPostP^{*} }
      \]

    \item if the constraint is
      \[
      \PostCon{p'}{\gamma'}{q} \Extend f(r) \Leq \PostCon{p}{\gamma}{q}
      \]
      or
      \[
      \PostCon{q''}{\gamma'}{q} \Extend \PostCon{p'}{\epsilon}{q''}
                           \Extend f(r) \Leq \PostCon{p}{\gamma}{q}
      \]
      where $r$ must be
      $r = \Ang{p', \gamma'} \PdsRule{} \Ang{p, \gamma} \in \PdsRules$.
      The case is analogous to the previous one (we just have $\gamma$ instead
      of $\epsilon$).

    \item if the constraint is one of
      \[
      \One \Leq \PostCon{p}{\gamma_1}{q_{p, \gamma_1}}
      \]
      or
      \[
      \PostCon{q''}{\gamma'}{q} \Extend \PostCon{p'}{\epsilon}{q''}
                           \Extend f(r) \Leq \PostCon{q_{p, \gamma_1}}{\gamma_2}{q}
      \]
      (alternatively without the $\epsilon$-transition:
      \[
      \PostCon{p'}{\gamma'}{q} \Extend f(r) \Leq \PostCon{q_{p, \gamma_1}}{\gamma_2}{q}
      \]
      but we will only consider the former, since it is a bit more complex and
      the proof for the latter is almost the same). \\
      We know that
      $r = \Ang{p', \gamma'} \PdsRule{} \Ang{p, \gamma_1 \gamma_2} \in \PdsRules$
      and so that we have $r_1, r_2 \in \PdsRulesPostP$ such that
      $r_1 = \Ang{p', \gamma'} \PdsRule{} \Ang{q_{p, \gamma_1}, \gamma_2}$
      and
      $r_2 = \Ang{q_{p, \gamma_1}, \epsilon} \PdsRule{} \Ang{p, \gamma_1}$
      with $f(r_1) = f(r)$ and $f(r_2) = \One$.

      For the first trivial inequality the property is clearly preserved. Let us
      focus on the second one. We know that
      \begin{align}
      \PostSolC{_i^{j+1}}{q_{p, \gamma_2}}{\gamma_2}{q}
        & =        \PostSolC{_i^j}{q_{p, \gamma_2}}{\gamma_2}{q} \notag \\
        & \Combine \left( \PostSolC{_i^j}{q'}{\gamma'}{q}
          \Extend  \PostSolC{_i^j}{p'}{\epsilon}{q'}
          \Extend f(r) \right)
        \label{lem:ann/post-precise-1}
      \end{align}
      for some $q' \in Q$.
      Using induction hypothesis we have that:
      \begin{align}
      \PostSolC{_i^j}{q_{p, \gamma_1}}{\gamma_2}{q} & \Leq \BigCombine \Set{
        v(\sigma) \mid    \Ang{q, \epsilon}
        \PdsTransS{\sigma} \Ang{q_{p, \gamma_1}, \gamma_2}
      } \label{lem:ann/post-precise-2} \\
      \PostSolC{_i^j}{q'}{\gamma'}{q} & \Leq \BigCombine \Set{
        v(\sigma) \mid    \Ang{q, \epsilon}
        \PdsTransS{\sigma} \Ang{q', \gamma'}
      } \notag \\
      \PostSolC{_i^j}{p'}{\epsilon}{q'} & \Leq \BigCombine \Set{
        v(\sigma) \mid    \Ang{q', \epsilon}
        \PdsTransS{\sigma} \Ang{p', \epsilon}
      } \notag
      \end{align}
      Using the last two and the fact that the flow algebra is affine, we get the
      following
      \begin{align*}
      & \PostSolC{_i^j}{q'}{\gamma'}{q}
        \Extend \PostSolC{_i^j}{p'}{\epsilon}{q'}
        \Extend f(r) \\
      & \Leq \Set{ v(\sigma_1) \Extend v(\sigma_2) \Extend f(r) \mid
        \Ang{q, \epsilon}  \PdsTransS{\sigma} \Ang{q', \gamma'},
        \Ang{q', \epsilon} \PdsTransS{\sigma} \Ang{p', \epsilon}
      } \\
      & \Leq \Set{ v(\sigma) \Extend f(r_1) \Extend f(r_2) \mid
        \Ang{q, \epsilon} \PdsTransS{\sigma} \Ang{p', \gamma'}
        \PdsTrans{r_1} \Ang{q_{p, \gamma_1}, \gamma_2}
        \PdsTrans{r_2} \Ang{p, \gamma_1 \gamma_2}
      } \\
      & \Leq \Set{ v(\sigma) \mid
        \Ang{q, \epsilon} \PdsTransS{\sigma} \Ang{p, \gamma_1 \gamma_2}
      }
      \end{align*}
      So from this and \eqref{lem:ann/post-precise-1} and
      \eqref{lem:ann/post-precise-2} we have the desired result.
    \end{itemize}
  \end{description}
\end{description}
\qed
\end{proof}

\subsection{Proof of Lem.~\ref{lem:ann/post-completeness-path}}
\label{app:ann/post-completeness-path}

Consider a weighted pushdown system $\WPds = (\Pds, \FlowAlgebra, f)$ where
$\FlowAlgebra$ is affine and a $\APostStarC$ automaton created by the saturation
procedure.
For every path $\rho = q' \AutoTransR{s} q$
($s \in \Gamma^{*}$) in this automaton we have that
\[
\PostSolP{q}{s}{\rho}{q'} \Leq \BigCombine \Set{ v(\sigma) \mid
  \Ang{q', \epsilon} \PdsTransS{\sigma} \Ang{q, s}, \sigma \in \PdsRulesPostP^{*} }
\]

\begin{proof}
The proof will proceed with the induction on the number of transitions in
$\rho$ (we will use the inductive definition of $\PostSolSymb$).
\begin{description}
\item[$|\rho| = 1$] According to the definition of $\PostSolSymb$ we have
  \[
  \PostSolP{q}{s}{\rho}{q'} = \PostSol{q}{\gamma_\epsilon}{q'}
  \]
  The result follows from Lemma \ref{lem:ann/post-completeness-transition}.
\item[$|\rho| > 1$] Again using the definition of $\PostSolSymbP$ we have
  \[
  \PostSolP{q}{s}{\rho}{q'} =
    \PostSolP{q''}{s'}{\rho'}{q'} \Extend \PostSol{q}{\gamma_\epsilon}{q''}
  \]
  where $s = \gamma_\epsilon s'$, $q'' \in Q$, and
  \[
  \rho = \underbrace{q' \AutoTransRS{s'} q''}_{\rho'} \AutoTransR{\gamma_\epsilon} q
  \]
  Now we can use the Lemma \ref{lem:ann/post-completeness-transition} along with
  the induction hypothesis (since $|\rho| > |\rho'|$) to get:
  \begin{align*}
  \PostSol{q}{\gamma_\epsilon}{q''}
    & \Leq
    \BigCombine \Set{ v(\sigma) \mid
      \Ang{q'', \epsilon} \PdsTransS{\sigma} \Ang{q, \gamma_\epsilon}, \sigma
      \in \PdsRulesPostP
    } \\
  \PostSolP{q''}{s'}{\rho'}{q'}
    & \Leq
    \BigCombine \Set{ v(\sigma) \mid
      \Ang{q', \epsilon} \PdsTransS{\sigma} \Ang{q'', s'}, \sigma \in
      \PdsRulesPostP
    }
  \end{align*}
  Finally, we use the fact that the flow algebra is affine:
  \begin{align*}
  & \PostSolP{q}{s}{\rho}{q'} \\
    & \Leq
    \BigCombine \Set{ v(\sigma) \Extend v(\sigma') \mid
      \Ang{q', \epsilon} \PdsTransS{\sigma'} \Ang{q'', s'},
      \Ang{q'', \epsilon} \PdsTransS{\sigma} \Ang{q, \gamma_\epsilon},
      \sigma, \sigma' \in \PdsRulesPostP
    } \\
    & \Leq
    \BigCombine \Set{ v(\sigma) \Extend v(\sigma') \mid
      \Ang{q', \epsilon} \PdsTransS{\sigma} \Ang{q, s}, \sigma \in
      \PdsRulesPostP
    }
  \end{align*}
\end{description}
\qed
\end{proof}

\subsection{Proof of Thm.~\ref{thm:ann/post-completeness}}
\label{app:ann/post-completeness}

Consider an automaton $\APostStarC$ constructed by the saturation procedure and
the least solution $\PostSolSymb$ to the set of its constraints $\Constraints$.
If the flow algebra is affine then for every path $\rho = q_f \AutoTransRS{s} p$
where $q_f \in F$ we have that
\[
\PostSolP{p}{s}{\rho}{q_f} = \BigCombine \Set{ v(\sigma) \mid
  \Ang{q_f, \epsilon} \PdsTransS{\sigma} \Ang{p, s}, \sigma \in \PdsRulesPostP^{*}
}
\]

\begin{proof}
Follows directly from Theorem~\ref{thm:ann/post-soundness}
and Lemma~\ref{lem:ann/post-completeness-path}.
\qed
\end{proof}